\documentclass[journal]{IEEEtran}

\usepackage{amssymb}
\usepackage{amsmath,amssymb,amsfonts}
\usepackage{multirow}
\usepackage{graphicx}
\usepackage{textcomp}
\usepackage{amsthm}

\usepackage{setspace}
\usepackage{subfigure}

\usepackage{graphicx}
\usepackage{algorithm}
\usepackage{algorithmic}
\usepackage{amsmath}
\usepackage{color}
\usepackage{amsfonts}

\ifCLASSINFOpdf
\else
\fi

\newtheorem{thm}{Theorem}
\newtheorem{lem}{Lemma}

\newtheorem{exmp}{Example}

\newtheorem{pro}{Problem}

\begin{document}

\title{Overall Evaluations on Benefits of Influence When Disturbed by Rivals}

\author{Jianxiong Guo,
	Yapu Zhang,
	Weili Wu,~\IEEEmembership{Senior Member,~IEEE}
	\thanks{J. Guo and W. Wu are with the Department
		of Computer Science, Erik Jonsson School of Engineering and Computer Science, Univerity of Texas at Dallas, Richardson, TX, USA; Y. Zhang is with the School of Mathematical Sciences, University of Chinese Academy of Sciences, Beijing, CHN.
		
		E-mail: jianxiong.guo@utdallas.edu}
	\thanks{Manuscript received April 19, 2005; revised August 26, 2015.}}

\markboth{Journal of \LaTeX\ Class Files,~Vol.~14, No.~8, August~2015}%
{Shell \MakeLowercase{\textit{et al.}}: Bare Demo of IEEEtran.cls for IEEE Journals}

\maketitle

\begin{abstract}
	Influence maximization (IM) is a representative and classic problem that has been studied extensively before. The most important application derived from the IM problem is viral marketing. Take us as a promoter, we want to get benefits from the influence diffusion in a given social network, where each influenced (activated) user is associated with a benefit. However, there is often competing information initiated by our rivals diffusing in the same social network at the same time. Consider such a scenario, a user is influenced by both my information and my rivals' information. Here, the benefit from this user should be weakened to certain degree. How to quantify the degree of weakening? Based on that, we propose an overall evaluations on benefits of influence (OEBI) problem. We prove the objective function of the OEBI problem is not monotone, not submodular, and not supermodular. Fortunately, we can decompose this objective function into the difference of two submodular functions and adopt a modular-modular procedure to approximate it with a data-dependent approximation guarantee. Because of the difficulty to compute the exact objective value, we design a group of unbiased estimators by exploiting the idea of reverse influence sampling, which can improve time efficiency significantly without losing its approximation ratio. Finally, numerical experiments on real datasets verified the effectiveness of our approaches regardless of performance and efficiency.
\end{abstract}


\IEEEpeerreviewmaketitle

\section{Introduction}
\IEEEPARstart{T}{he} online social media, such as Twitter, Facebook, Wechat, and LinkedIn, were booming prosperously in the recent decade and become a dominating method to contact with others and make friends \cite{zhang2017situational}. People are more inclined to share their comments about some hot issues at every moment in these platforms. By the end of December 2019, there are more than 3.725 billon users active in these social media. The relationships among the users on these social platforms can be denoted by social networks. A large number of messages can be shared rapidly over the networks. Subsequently, influence maximization (IM) \cite{kempe2003maximizing} was formulated to focus on a problem that selects a small subset of users (seed set) for an information cascade to maximize the expected follow-up adoptions (influence spread). It is a natural generalization for viral marketing. The IM problem was based on the two influence diffusion models, independent cascade model (IC-model) and linear threshold model (LT-model), and they can be summarized into the trigger model. Besides, they \cite{kempe2003maximizing} proved the expected influence spread is monotone and submodular, thereby a $(1-1/e)$-approximation can be obtained by the greedy algorithm implemented by the Monte-Carlo (MC) simulations.

Since this seminal work, it derives a series of optimization problems, such as profit maximization (PM) \cite{lu2012profit} \cite{dong2017dynamics} \cite{9089295}, competitive IM \cite{bharathi2007competitive} \cite{guo2019novel}, and rumor blocking \cite{tong2017efficient} \cite{guo2019multi}. Consider us as a promoter to initiate an information cascade, we aim to get benefits from the influence spread started from our selected seed set in a social network. If a user is activated during the influence diffusion, we can get a benefit associated with her. Suppose it exists cost needed to pay when selecting a seed set, the profit is defined by the total benefits of influence spread minus the cost of this seed set, where the PM problem aims to maximize the expected profit. However, this is only an idealized state, where there is no competitor diffusing its cascade simultaneously. Generally, more than one type of information can flood the same network. In the competitive IM problem, there are multiple information cascades diffusing their respective influence independently, where it assumes a user can only be activated by one cascade successfully. It aims to select a seed set to maximize our own expected influence spread or to minimize the influence spread from other competing cascades (rumor blocking).

Combining the PM and competitive IM problem together, it formulates the competitive PM problem that maximizes our own expected profit when there are multiple information cascades. However, this model has a crucial drawback because each user can only be activated by one cascade. Actually, for a user in a social network, she may be influenced by multiple cascades from different promoters. If a user is activated by our cascade but activated by rivals' cascades contemporarily, the benefit we can get from her will be weakened, even be negative. Let us consider the following example.

\begin{exmp}
	Take us as an Apple carrier, we want to popularize a new iPhone across a given network by influence diffusion. If a user is influenced by us, we can get a benefit from her according to her appraisal about our product. When there is a rival, such as Samsung, existing, it will promote its phone by diffusing the influence as well. If a user is influenced by both Samsung and us, its appraisal about our product is very likely to be reduced after comparing it with Samsung. The benefit associated with her will be reduced even to be negative.
\end{exmp}

Based on this realistic scenario, we propose an overall evaluations on benefit of influence (OEBI) problem, where we define how to quantify and maximize the benefits of influence because of the rival's disturbance. We show that the OEBI problem is NP-hard and its objective function is not monotone, not submodular, and not supermodular. Because there is no direct approach to approximate it with a theoretical bound, we decompose this objective function into the difference of two monotone and submodular functions. Then, we adopt a modular-modular procedure \cite{iyer2012algorithms} that replaces the first submodular function with one of its lower bound and the second submodular function with one of its upper bound. Then, a data-dependent approximation ratio can be obtained by this procedure. Moreover, it is \#P-hard to compute the exact objective value under the IC-model \cite{chen2010scalable} and LT-model \cite{chen2010scal}. Even though we can estimate our objective value by use of MC simulations, the terrible time inefficiency is unavoidable, which restricts its scalability to larger networks. Based on the idea of reverse influence sampling (RIS) \cite{borgs2014maximizing}, we design a group of unbiased estimators to estimate the value of our objective function. If the number of samplings is large enough, its estimation error is neglectable. Next, we take this estimator as the input of modular-modular procedure, which reduces the running time greatly while maintaining the approximation guarantee. Finally, we conduct several experiments to evaluate the superiority of our proposed method to other heuristic algorithms, where they support the effectiveness and efficiency of our method strongly.

\textbf{Organization:} Sec. \uppercase\expandafter{\romannumeral2} surveys the-state-of-art works. Sec. \uppercase\expandafter{\romannumeral3} is dedicated to introduce diffusion model, background, and define the OEBI problem formally. The monotonicity, submodularity, and computability are presented in Sec. \uppercase\expandafter{\romannumeral4}. Sec. \uppercase\expandafter{\romannumeral5} is the main contributions, including algorithm design, sampling techniques, and approximation guarantee. Numerical experiments and performance analysis are presented in Sec. \uppercase\expandafter{\romannumeral7} and \uppercase\expandafter{\romannumeral8} is the conclusion for this paper. 

\section{Related Works}
\textbf{Influence Maximization:} Kempe \textit{et al.} \cite{kempe2003maximizing} came up with the IC-model and LT-model, formulated IM problem as a monotone submodular maximization problem, and gave a greedy algorithm that achieves $(1-1/e-\varepsilon)$-approximation implemented by MC simulations. Chen \textit{et al.} proved it is \#P-hard to compute the expected influence spread given a seed set under the IC-model \cite{chen2010scalable} and LT-model \cite{chen2010scal}. Besides, they devised two efficient heuristic algorithms to solve the IM problem and evaluate their scalability. Contemporarily, a series of heuristic algorithms emerged, such as cost-effective lazy forward strategy \cite{leskovec2007cost} and degree discount heuristics \cite{chen2009efficient}. Brogs \textit{et al.} \cite{borgs2014maximizing} made a breakthrough. They proposed the concept of RIS to estimate the expected influence spread, which is scalable in practice and has a theoretical bound at the same time. Then, a series of researchers designed more efficient algorithms that achieve $(1-1/e-\varepsilon)$-approximation based on the RIS. Tang \textit{et al.} \cite{tang2014influence} \cite{tang2015influence} proposed TIM/TIM+ algorithms first and then develop a more efficient IMM based on the martingale analysis. Besides, it was improved further by SSA/DSSA \cite{nguyen2016stop} and OPIM \cite{tang2018online}.

\textbf{Competitive IM and Profit Maximization:}
Bharathi \textit{et al.} \cite{bharathi2007competitive} studied the competitive IM first and generalized it as a game of influence diffusion with multiple competing cascade. Lu \textit{et al.} \cite{lu2015competition} created a comparative IC-model that includes all settings of influence propagation from competition to complementarity. Tong \textit{et al.} \cite{tong2019multi} proposed an independent multi-cascade model and studied a multi-cascade IM problem under this model systematically, where they designed efficient algorithm and obtained a data-dependent approximation guarantee. In the classic PM problem \cite{lu2012profit} \cite{tang2016profit}, they usually considered the cost of a seed set is modular with respect the seed node in this seed set, which implies the profit function is still submodular but not monotone. It can be generalized as the unconstrained submodular maximization problem, which can be addressed by the double greedy algorithm within $(1/3)$-approximation and randomized double greedy algorithm within $(1/2)$-approximation \cite{buchbinder2015tight}. Tong \textit{et al.} \cite{tong2018coupon} considered the coupon allocation in the PM problem, and designed efficient randomized algorithms to achieve $(1/2-\varepsilon)$-approximation with high probability. Guo \textit{et al.} \cite{8952599} proposed a budgeted coupon problem whose domain is constrained and provided a continuous double greedy algorithm with a valid approximation. However, in our model, the formulation of competitiveness and definition of benefit are different from one of the above works.

\textbf{Non-submodular Maximization:} However, many realistic problems derived from the IM do not satisfy the submodularity. For a monotone non-submodular function, we can use the supermodular degree \cite{feige2013welfare} and curvature \cite{wang2016approximation} to analyze the approximation of greedy algorithm to maximize it. Then, Lu \textit{et al.} \cite{lu2015competition} devised a sandwich approximation framework, which can obtain a data-dependent approximation ratio by maximizing its submodular upper and submodular lower bounds, the return the solution that can maximize the original objective function as the final result. However, our objective function of the OEBI problem is not monotone. For a non-monotone non-submodular function, it can be decomposed into the difference of two submodular functions \cite{narasimhan2005submodular}, which can be approximated effectively by the submodular-supermodular procedure \cite{narasimhan2005submodular} and modular-modular procedure \cite{iyer2012algorithms}. In this paper, we design an efficient randomized algorithm to solve our OEBI problem with a satisfactory approximation guarantee based on the RIS and modular-modular procedure.

\begin{figure*}[!t]
	\centering
	\subfigure[The initial state]{
		\includegraphics[width=0.24\linewidth]{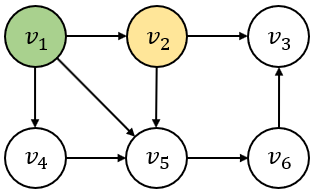}
	}%
	\subfigure[A realizating $g\sim\Omega^p$]{
		\includegraphics[width=0.24\linewidth]{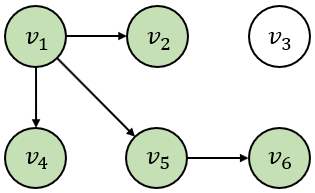}
	}%
	\subfigure[A realization $g'\sim\Omega^r$]{
		\includegraphics[width=0.24\linewidth]{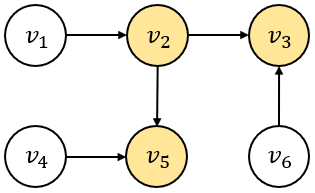}
	}%
	\subfigure[The final state]{
		\includegraphics[width=0.24\linewidth]{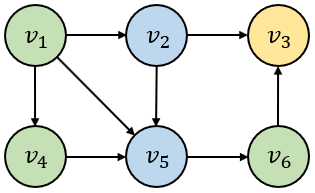}
	}%
	\centering
	\caption{This is an example to demonstrate the diffusion precess caused by a positive cascade and a negative cascade, where the green nodes, yellow nodes, and blue nodes are activated by the positive cascade, rival cascade, and both positve and rival cascades.}
	\label{fig1}
\end{figure*}

\section{Problem Formulation}
In this section, we introduce the diffusion model first and then formulate the OEBI problem.

\subsection{Diffusion Model and Realization}
	Let $G=(V,E)$ be a directed graph that represents a social network where where $V=\{v_1,v_2,\cdots,v_n\}$ is the set of $n$ users, $E=\{e_1,e_2,\cdots,e_m\}$ is the set of $m$ directed edges. For each directed edge $(u,v)\in E$, it models their friendship where $u$ (resp. $v$) is an incoming neighbor (resp. outgoing neighbor) of $v$ (resp. $u$). Moreover, the set of incoming neigbhbors (resp. outgoing neighbors) of node $u\in V$ is denoted by $N^-(v)$ (resp. $N^+(v)$).

	Given a seed set $S\subseteq V$, the influence diffusion model is a discrete-time stochastic process started from the seed nodes in $S$. In the beginning, all nodes in the seed set $S$ are active, but the other nodes are inactive. At time step $t_i$, we denote by $S_i$ the current active node set. Thereby we have $S_0:=S$ at $t_0$. Under the IC-model \cite{kempe2003maximizing}, there is a diffusion probability $p_{uv}\in (0,1]$ associated with each edge $(u,v)\in E$. At time step $t_i$ for $i\geq 1$, we have $S_i:=S_{i-1}$ first; then, each new activated node $u\in(S_{t-1}\backslash S_{t-2})$ in the last time step has one chance to activate its each inactive outgoing neighbor $v$ with the probability $p_{uv}$. We add $v$ into $S_i$ if $u$ activates $v$ successfully. The influence diffusion stops when no node can be activated further. The problems we will discuss in the subsequent sections are defaulted on the IC-model, but they can be extended to other influence models easily.
	
	Here, a specific IC-model based on graph $G$ can be defined as $\Omega=(G,P)$ where $P=\{p_{e_1},p_{e_2},\cdots,p_{e_m}\}$ is the set of $m$ edge probabilities. Given a specific IC-model $\Omega$, we define $g\sim \Omega$ as a realization sampled from $\Omega$, which is an instance of influence diffusion on this probabilistic graph. Under the IC-model, a realization is residual graph built by removing each edge $(u,v)\in E$ with probability $1-p_{uv}$. Thereby we have $\Pr[g]=\prod_{e\in E(g)}p_e\prod_{e\in E(G)\backslash E(g)}(1-p_e)$ and there is $2^m$ potential realizations in total.

	Given a seed set $S\subseteq V$ and a realization $g$, we denote by $I_g(S)$ the set of nodes that can be reachable from at least one node in this seed set. Thus, the expected number of active nodes over all potential realizations (expected influence spread) can be expressed as
	\begin{equation}
		\sigma_\Omega(S)=\mathbb{E}_{g\sim \Omega}\left[|I_g(S)|\right]=\sum_{g\in\mathcal{G}(\Omega)}\Pr[g]\cdot|I_g(S)|
	\end{equation}
	where $\mathcal{G}(\Omega)$ is the collection of all possible realizations sampled from $\Omega$. The IM problem is to select a seed set $S\subseteq V$ where $|S|\leq k$ such that the expected influence spread $\sigma(S)$ can be maximized. Given a set function $h:2^V\rightarrow\mathbb{R}$ and any two sets $S,T\subseteq V$, it is monotone if $h(S)\leq h(T)$ when $S\subseteq T\subseteq V$, submodular if $h(S\cup\{u\})-h(S)\geq h(T\cup\{u\})-h(T)$ when $S\subseteq T\subseteq V$ and $u\notin T$, and supermodular if $h(S\cup\{u\})-h(S)\leq h(T\cup\{u\})-h(T)$ when $S\subseteq T\subseteq V$ and $u\notin T$. Based on that, we have the expected influence spread $\sigma(\cdot)$ is monotone non-decreasing and submodular under the IC-model \cite{kempe2003maximizing}.
	
\subsection{Problem Definition}
	Consider a company, it wants to promote its new product by starting a cascade diffusing over the social network. Obviously, the expected influence spread is the benefit it can obtain. However, this is only in an ideal world because it does not consider whether there is the other cascade representing a competing product started by a rival company that diffuses over the social network at the same time. Thus, we can no longer evaluate this company's benefit only by the expected influence spread due to the rival's disturbance.
	
	Given a social network $G=(V,E)$, there are multiple cascades diffusing on this network simultaneously. A user is referred as $C$-active if she is activated by cascade $C$. Consider such a scenario, we define a positive cascade $C_p$ which represents the influence diffusion for the new product we want to promote over the network. It exists a rival cascade $C_r$ represents the influence diffusion for a competing product started by some rival company. Now, due to the existence of this competing cascade, our benefit from the influence spread of cascade $C_p$ will be disturbed and impaired to some extent. Given a rival seed set $S_r$, we need to find a positive seed set $S_p$ and start this positive cascade such that it can avoid the negative effects of the rival cascade started from $S_r$ as much as possible.
	
	Next, we discuss how to quantify the disturbance caused by the rival cascade to our benefit. Given a social network $G=(V,E)$, we consider a positive cascade $C_p$ diffuses under the IC-model $\Omega^p=(G,PP)$ and a rival cascade $C_r$ diffuses under the IC-model $\Omega^r=(G,PR)$, where $PP$ (resp. $PR$) is an edge probability distribution of $\Omega^p$ (resp. $\Omega^r$). These two cascades diffuse over the network $G$ respectively and independently. Then, we suppose each node $u\in V$ is associated with a benefit weight $p(u)\in\mathbb{R}_+$, which implies the benefit can be obtained from the fact that $u$ is $C_p$-active but not $C_r$-active. In other words, it is the earning from activating user $u$ by our positive cascade but not activating it by the rival cascade. Moreover, we suppose each node $u\in V$ is associated with a disturbed benefit weight $q(u)\in\mathbb{R}$ with $q(u)\leq p(u)$, which implies the earning can be obtained from the fact that $u$ is $C_p$-active and $C_r$-active. Here, the disturbed benefit weight describes the degree of disturbance caused by the rival cascade. For a user $u\in V$, her degree of disturbance caused by the rival cascade rests with its disturbed benefit weight $q(u)$. If $q(u)\in[0,p(u)]$, it means that the rival cascade will not cause a negative effect on this node $u$ even though it cuts down the benefit can be obtained from activating this node by positive cascade. If $q(u)\in(-\infty,0)$, it means that the rival cascade will cause a negative effect on this node. Thus, this $q$ controls the degree of disturbance caused by the rival cascade.
	
	Given a rival seed set $S_r\subseteq V$, the expected overall benefit from our positive seed set $S_p$ can be defined as
	\begin{flalign}
		f(S_p)&=\mathbb{E}_{g\sim\Omega^p}\mathbb{E}_{g'\sim\Omega^r}[f_{g,g'}(S_p)]\\
		&=\sum_{g\in\mathcal{G}(\Omega^p)}\Pr[g]\sum_{g'\in\mathcal{G}(\Omega^r)}\Pr[g']\cdot f_{g,g'}(S_p)
	\end{flalign}
	where $f(S_p)$ is the expectation over the realizations sampled from the IC-model $\Omega^p$ and $\Omega^n$. Given the two realizations $g\sim\Omega^p$ and $g'\sim\Omega^r$, the overall benefit of influence diffusion can be defined as
	\begin{equation}
		f_{g,g'}(S_p)=\sum_{u\in I_g(S_p)\backslash I_{g'}(S_r)}p(u)+\sum_{u\in I_g(S_p)\cap I_{g'}(S_r)}q(u)
	\end{equation}
	where the first term is the benefit from nodes activated only by $C_p$ and the second term is the disturbed benefit from nodes activated by both $C_p$ and $C_r$.
	
	Let us look at an example shown in Fig. \ref{fig1}. Shown as Fig. \ref{fig1}(a), the positive seed set is $S_p=\{v_1\}$ and the rival seed set $S_r=\{v_2\}$ in the beginning. Then, the influence spread started from $S_p$ is shown as Fig. \ref{fig1}(b), which is a realization sampled from its IC-model $\Omega^p$. Similarly, the influence spread started from $S_r$ is shown as Fig. \ref{fig1}(c), which is a realization sampled from its IC-model $\Omega^r$. From here, we can see that they diffuse respectively and independently. Finally, node $v_2$ and $v_5$ are activated by both the positive and rival cascades, thereby we have $I_g(S_p)\cap I_{g'}(S_r)=\{v_2,v_5\}$ shown as Fig. \ref{fig1}(d). Therefore, we have the overall benefit under this realization is $f_{g,g'}(S_p)=p(v_1)+p(v_4)+p(v_6)+q(v_2)+q(v_5)$. The overall evaluations on benefit of influence (OEBI) problem is
	\begin{pro}[OEBI]
		Given a social network $G=(V,E)$, a rival seed set $S_r$, and a budget $k$, the OEBI problem is aimed at finding a positive set set $S_p\subseteq V$, where $|S_p|\leq k$, such that its expected overall benefit $f(S_p)$ can be maximized, that is $S_p^*=\arg\max_{|S_p|\leq k}f(S_p)$.
	\end{pro}

\section{Further Discussions about OEBI}
	In this section, we analyze the properties of OEBI first and introduce how to decompose its objective function.
\subsection{The Properties}
	Given the rival seed set $S_r=\emptyset$, the OEBI problem can be reduced to the classical IM problem if we assume $p(u)=1$ for each $u\in V$. Thus, the OEBI problem is NP-hard through inheriting the NP-hardness of IM problem \cite{kempe2003maximizing} under the IC-model. Moreover, it is \#P-hard to compute the expected overall benefit because of the \#P-hardness to compute the expected influence spread under the IC-model \cite{chen2010scalable}. Next, we will analyze the monotonicity, submodularity, and supermodularity of the expected overall benefit function $f(S_p)$ with respect to $S_p$ step by step.
	\begin{thm}
		The objective function of the OEBI problem $f(S_p)$ is not monotone with respect to $S_p$.
	\end{thm}
	\begin{proof}
	 	We consider the simplest case where the graph $G$ has only one node. Here, we have $V=\{v\}$ ands $E=\emptyset$. Given a rival seed set $S_r=\{v\}$, the expected overall benefit $f(\{v\})=q(u)$ and $f(\emptyset)=0$. Subsequently, we have $f(\{v\})-f(\emptyset)\geq 0$ if $q(u)\geq 0$; and $f(\{v\})-f(\emptyset)\leq 0$ if $q(u)\leq 0$. Thus, the monotonicity of $f(S_p)$ depends on the definition of dusturbed earning weights.
	\end{proof}

	\begin{thm}
		The objective function of the OEBI problem $f(S_p)$ is not submodular with respect to $S_p$ and not supermodular with respect to $S_p$.
	\end{thm}
	\begin{proof}
	    Take a counterexample to prove it, we assume $p=p(u)$ and $q=q(u)$ for each node $u\in V$ with $q\in(-\infty,-p)$. Shown as Fig. \ref{fig2}, we can see that $f(\{v_2,v_4\})=2p-q$ and $f(\{v_1,v_4\})=5p-q$. First, we have $f(\{v_2,v_4\})-f({v_4})=p+q<f(\{v_1,v_2,v_4\})-f({v_1,v_4})=0$, thereby $f(S_p)$ is not submodular with respect to $S_p$. Then, we have $f(\{v_4,v_5\})-f({v_4})=2p>f(\{v_1,v_4,v_5\})-f({v_1,v_4})=0$, thereby $f(S_p)$ is not supermodular with respect to $S_p$. 
	\end{proof}
	\begin{figure}[h]
		\centering
		\subfigure[$S_p=\{v_2,v_4\}$]{
			\includegraphics[width=0.48\linewidth]{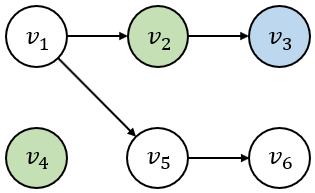}
		}%
		\subfigure[$S_p=\{v_1,v_4\}$]{
			\includegraphics[width=0.48\linewidth]{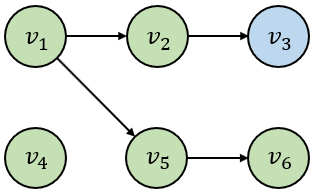}
		}%
		\centering
		\caption{This is an example to demonstrate the submodularity and supermodularity in Theorem 2.}
		\label{fig2}
	\end{figure}

\subsection{Decomposition of Our Objective Function}
	From the above subsection, the expected overall benefit is non-monotone, non-submodular, and non-supermodular, therefore, it is hard to get an effective solution with an approximation ratio. Narasimhan \textit{et al.} \cite{narasimhan2005submodular} proposed a DS decomposition, which pointed out any set function can be decomposed into the difference of two submodular set functions. Even that, whether such two submodular set functions can be found in polynomial time is still unknown. Look at the (4), the overall benefit $f_{g,g'}(S_p)$ under the $g\sim\Omega^p$ and $g'\sim\Omega^r$ can be re-arranged as
	\begin{equation}
		f_{g,g'}(S_p)=\sum_{u\in I_g(S_p)}p(u)-\sum_{u\in I_g(S_p)\cap I_{g'}(S_r)}(p(u)-q(u))
	\end{equation}
	Thus, we can decompose the expected overall benefit as $f(S_p)=w(S_p)-z(S_p)$, where $w(S_p)$ and $z(S_p)$ are defined as follows, that is
	\begin{flalign}
		&w(S_p)=\mathbb{E}_{g\sim\Omega^p}\left[\sum\nolimits_{u\in I_g(S_p)}p(u)\right]\\
		&z(S_p)=\mathbb{E}_{g\sim\Omega^p}\mathbb{E}_{g'\sim\Omega^r}\left[\sum\nolimits_{u\in I_g(S_p)\cap I_{g'}(S_r)}l(u)\right]
	\end{flalign}
	where we denote $l(u)=p(u)-q(u)$. Similarly, we denote $w_g(S_p)=\sum_{u\in I_g(S_p)}p(u)$ under the $g\sim\Omega^p$ and $z_{g,g'}(S_p)=\sum_{u\in I_g(S_p)\cap I_{g'}(S_r)}l(u)$ under the $g\sim\Omega^p$ and $g'\sim\Omega^r$.
	\begin{thm}
		The function $w(S_p)$ is monotone non-decreasing and submodular with respect to $S_p$.
	\end{thm}
	\begin{proof}
		The function $w(S_p)$ is the objective function of weighted IM problem. It can be reduced to weighted maximum set cover problem, which is monotone non-decreasing and submodular since $p(u)\geq 0$ for any $u\in V$.
	\end{proof}
	\begin{thm}
		The function $z(S_p)$ is monotone non-decreasing and submodular with respect to $S_p$.
	\end{thm}
	\begin{proof}
		Given a rival seed set $S_r$, realization $g\sim\Omega^p$, and $g'\sim\Omega^r$, we consider the monotonicity and submodularity based on $z_{g,g'}(S_p)$. First, it is apparent that $z_{g,g'}(S_p)$ is monotone non-decreasing with respect to $S_p$. Then, there are two positive seed set $S_p^1$ and $S_p^2$ with $S_p^1\subseteq S_p^2$. For any node in $I_{g'}(S_r)$, if it is reachable from node $v$ but is not reachable from $S_p^2$, it must not be reachable from $S_p^1$ since $S_p^1\subseteq S_p^2$. Thereby we have $z_{g,g'}(S_p^1\cup\{v\})-z_{g,g'}(S_p^1)\geq z_{g,g'}(S_p^2\cup\{v\})-z_{g,g'}(S_p^2)$ because of $l(u)\geq 0$ for any $u\in V$, which implies that $z_{g,g'}(S_p)$ is submodular with respect to $S_p$. Besides, $y(S_p)$ is a linear combination of $z_{g,g'}(S_p)$, thus $z(S_p)$ is monotone non-decreasing and submodular.
	\end{proof}

	Therefore, the expected overall benefit $f(S_p)$ has been decomposed into the difference of two monotone submodular functions $w(S_p)$ and $z(S_p)$ definitely.
	
\section{Algorithm Degisn and Speedup}
	From the last section, our objective function is not monotone, not submodular, and not supermodular. Fortunately, it can be decomposed into the difference of two monotone submodular functions. Iyer \textit{et al.} \cite{iyer2012algorithms} proposed a modular-modular procedure to minimize the difference between two submodular functions approximately. First, we need to define the modular upper bound and modular lower bound for a given submodular function.
	
	\begin{algorithm}[!t]
		\caption{\text{Modular-modular}}\label{a1}
		\begin{algorithmic}[1]
			\renewcommand{\algorithmicrequire}{\textbf{Input:}}
			\REQUIRE A set function $f:2^V\rightarrow\mathbb{R}$
			\STATE Initialize: $X^t\leftarrow\emptyset$, $t\leftarrow0$
			\WHILE {$X^{t+1}\neq X^t$}
			\STATE Selects a permutation $\alpha^t$ that contains $X^t$ where the element in $X^t$ are ranked ahead
			\STATE $X^{t+1}\leftarrow\arg\max_{|Y|\leq k}\left\{h_{X^t,\alpha^t}^w(Y)-m_{X^t}^z(Y)\right\}$
			\STATE $t\leftarrow t+1$
			\ENDWHILE
			\RETURN $X^t$
		\end{algorithmic}
	\end{algorithm}
\subsection{Modular-modular Procedure}
	Given a submodular function $b(\cdot)$, it has two modular upper bounds based on a given set $X\subseteq V$, that is
	\begin{flalign}
		&m_{X,1}^b(Y)=b(X)-\sum_{j\in X\backslash Y}b(j|X\backslash j)+\sum_{j\in Y\backslash X}b(j|\emptyset)\\
		&m_{X,2}^b(Y)=b(X)-\sum_{j\in X\backslash Y}b(j|V\backslash j)+\sum_{j\in Y\backslash X}b(j|X)
	\end{flalign}
	where $b(S|T)=b(S\cup T)-b(T)$, $m_{X,1}^b(Y)\geq b(Y)$, and $m_{X,2}^b(Y)\geq b(Y)$. They are tight at set $X$, so we have $m_{X,1}^b(X)=m_{X,2}^b(X)=f(X)$.
	
	Given a set $X\subseteq V$, we define a permutation $\alpha$ of $V$ as $\alpha=\{\alpha(1),\alpha(2),\cdots,\alpha(n)\}$ where $\eta$'s chain contains $X$. Denote by $S_i^\alpha=\{\alpha(1),\alpha(2),\cdots,\alpha(i)\}$, we have $S_{|X|}^\alpha=X$, in other words, we put all the elements in $X$ prior to the elements in $V\backslash X$. Then, we define
	\begin{equation}
		h_{X,\alpha}^b(\alpha(i))=b(S_i^\alpha)-b(S_{i-1}^\alpha)
	\end{equation}
	where $h_{X,\alpha}^b(Y)=\sum_{v\in Y}h_{X,\alpha}^b(v)$ and $h_{X,\alpha}^b(Y)\leq b(Y)$ for any $Y\subseteq V$. Here, $h_{X,\alpha}^b(Y)$ is a lower bound of $b(Y)$. It is tight at set $X$, wo we have $h_{X,\alpha}^b(X)=b(X)$.
	
	From the (6) and (7), we adopt the modular-modular proceduce to solve it is formulated in Algorithm \ref{a1}.
	\begin{thm}
		The objective function $f(X^t)$ is monotone non-decreasing with respect to $t$. If the $h_{X^t,\alpha^t}^w(Y)-m_{X^t}^z(Y)$ in line 4 of Algorithm \ref{a1} reaches a local maximum under the $O(n)$ different permutations $\alpha^t$ and both upper bounds, then the $f(Y)$ is a local maximum.
	\end{thm}
	\begin{proof}
		Regardless of what the upper bound we use, at any round $t$, we have 		$f(X^{t+1})=w(X^{t+1})-z(X^{t+1})\geq h_{X^t,\alpha^t}^w(X^{t+1})-m_{X^t}^z(X^{t+1})\geq h_{X^t,\alpha^t}^w(X^{t})-m_{X^t}^z(X^{t})=w(X^{t})-z(X^{t})=f(X^{t})$ since the definitions of the upper and lower bounds and the tightness at set $X^{t}$.
		
		Suppose the Algorithm \ref{a1} converges at $X^{t+1}=X^t$, we consider the $O(n)$ different permutations $\alpha^t$ which are placed with different elements at position $\alpha^t(|X^t|)$ and $\alpha^t(|X^{t+1}|)$. First, we have $h_{X,\alpha}^w(S_i^{\alpha})=w(S_i^{\alpha})$, $m^z_{X^t,1}(X^t\backslash j)=z(X^t)-z(j|X^t\backslash j)=z(X^t\backslash j)$, and $m^z_{X^t,2}(X^t\cup j)=z(X^t)+z(j|X^t)=z(X^t\cup j)$. At the convergence, we have $h_{X^t,\alpha^t}^w(X^{t})-m_{X^t}^z(X^{t})\geq h_{X^t,\alpha^t}^w(Y)-m_{X^t}^z(Y)$ for any $Y\subseteq V$ under the $O(n)$ different permutations $\alpha^t$ and both upper bounds. Given a $\alpha^t$ with $\alpha^t(|X^t|)=i$ and $\alpha^t(|X^t|+1)=j$, we have $f(X^t)=w(X^t)-z(X^t)=h_{X^t,\alpha^t}^w(X^{t})-m_{X^t,1}^z(X^{t})\geq h_{X^t,\alpha^t}^w(X^{t}\backslash i)-m_{X^t,1}^z(X^{t}\backslash i)=f(X^t\backslash i)$ and $f(X^t)=w(X^t)-z(X^t)=h_{X^t,\alpha^t}^w(X^{t})-m_{X^t,2}^z(X^{t})\geq h_{X^t,\alpha^t}^w(X^{t}\cup j)-m_{X^t,1}^z(X^{t}\cup j)=f(X^t\cup j)$. Therefore, $f(X^t)$ is a local maximum at the convergence.
	\end{proof}

	\begin{algorithm}[!t]
		\caption{\text{ModularMax}}\label{a2}
		\begin{algorithmic}[1]
			\renewcommand{\algorithmicrequire}{\textbf{Input:}}
			\REQUIRE A permutation $\alpha^t$ and a set $X^t$
			\STATE Initialize: a map $unitValue=\{\}$
			\STATE Initialize: a set $X^{t+1}\leftarrow\emptyset$
			\STATE $zero\leftarrow h_{X^t,\alpha^t}^w(\emptyset)-m_{X^t}^z(\emptyset)$
			\FOR {each $u\in V$}
			\STATE $unitValue[u]\leftarrow h_{X^t,\alpha^t}^w(\{u\})-m_{X^t}^z(\{u\})-zero$
			\ENDFOR
			\FOR {$i=1$ to $k$}
			\STATE Select $u^*\in\max_{u\in V\backslash X^{t+1}}unitValue[u]$
			\IF {$unitValue[u^*]<0$}
			\STATE Break
			\ENDIF
			\STATE $X^{t+1}\leftarrow X^{t+1}\cup\{u^*\}$
			\ENDFOR
			\RETURN $X^{t+1}$
		\end{algorithmic}
	\end{algorithm}

	At each iteration in this algorithm, we need to maximize a modular function shown as in line 4 of Algorithm \ref{a1}, which can be implemented easily. For example, we can compute the objective value for each node $u\in V$ and then select all those which has a non-negative objective value. At the iteration $t$, given a permutation $\alpha^t$ and a set $X^t$, the algorithm that selects a set $Y$ where $|Y|\leq k$ to maximize the modular function $h_{X^t,\alpha^t}^w(Y)-m_{X^t}^z(Y)$ is shown in Algorithm \ref{a2}. The update rule in Algorithm \ref{a2} is according to $h(u|S)=h(u|T)=h(u|\emptyset)$ for any set $S,T\subseteq V$ if $h(\cdot)$ is a modular function.
	
	As for how to select a permutation $\alpha^t$ at each iteration $X^t$, the optimal solution is to select a permutation $\alpha^*$ such that $\alpha^t_*\in\arg\max_{\alpha^t}\max_{|Y|\leq k}\{h_{X^t,\alpha^t}^w(Y)-m_{X^t}^z(Y)\}$, however it is very difficult to execute. There are $n!$ permutations in total. Thus, a heuristic choice is to order the permutation $\alpha^t$ according to the magnititude of objective value for each node $u\in V$. We will compare the impact of different permutations on algorithm performance in later experiments.
	
	According to the (8) and (9), we have two upper bounds for a submodular function. Thereby the upper bound of the optimal value of our expected overall benefit $f(S_p^*)$ can be defined as follows:
	\begin{equation}
		\pi(X)=\max_{|Y|\leq k}\{\min\{m_{X,1}^w(Y),m_{X,2}^w(Y)\}-h^z_{X,\alpha}(Y)\}
	\end{equation}
	where $\min\{m_{X,1}^w(Y),m_{X,2}^w(Y)\}$ is aimed to make this upper bound tighter. It can be solved similar to the process of Algorithm \ref{a2}. Then, for any set $X$, we have $\pi(X)\geq\max_{|Y|\leq k}f(Y)$. Denote by $S_p^\circ$ the seed set returned by Algorithm \ref{a1}, we have $\pi(S_p^\circ)\geq f(S_p^*)$, then we are able to estimate the approximation ratio by $f(S_p^\circ)/\pi(S_p^\circ)$.
	
\subsection{Sampling Techniques}
	Given a seed set $S_p$, we adopt the technique of reverse influence sampling (RIS) to estimate $f(S_p)$ due to its \#P-hardness. Consider the IM problem under the IC-model $\Omega=(G,P)$, we introduce the concept of reverse reachable set (RR-set) first. A random RR-set $R$ can be generated by three steps: (1) selecting a node $u\in V$ uniformly; (2) sampling a realization $g\sim\Omega$; and (3) collecting those nodes in $g$ can reach $u$ and putting them into $R$. A RR-set rooted at node $u$ is a collection of nodes that are likely to influence $u$. A larger expected influence spread a seed set $S$ has, the higher the probability that $S$ intersects with a random RR-set is. Given a seed set $S$ and a random RR-set $R$, we have $\sigma_\Omega(S)=n\cdot\Pr[R\cap S\neq\emptyset]$.
	
	Back to our OEBI problem, the expected overall benefit can be denoted by $f(S_p)=w(S_p)-z(S_p)$. Thus, given a seed set $S_p$, we require to estimate $w(S_p)$ and $z(S_p)$ respectively. Here, we define $p(V)=\sum_{v\in V}p(v)$ and $l(V)=\sum_{v\in V}l(v)$ respectively for convenience. For the $w(S_p)$, a random RR-set $R_w$ can be generated by (1) selecting a node $u\in V$ with probability $p(u)/p(V)$; (2) sampling a realization $g\sim\Omega^p$; and (3) putting those nodes in $g$ can reach $u$ into $R_p$. Given a seed set $S_p$ and a random RR-set $R_w$, we have $w(S)=p(V)\cdot\Pr[R_w\cap S_p\neq\emptyset]$. For the $z(S_p)$, a random RR-set $R_z$ can be generated by (1) selecting a node $u\in V$ with probability $l(u)/l(V)$; (2) sampling a realization $g\sim\Omega^p$ and a realization $g'\sim\Omega^r$ independently; and (3) putting those nodes in $g$ can reach $u$ into $R_{z,1}$ and those nodes in $g'$ can reach $u$ into $R_{z,2}$ where $R_z=(R_{z,1},R_{z,2})$.
	\begin{lem}
		Given a seed set $S_p$, a rival seed set $S_r$, and a random RR-set $R_z=(R_{z,1},R_{z,2})$, we have
		\begin{equation}
			z(S_p)=l(V)\cdot\Pr\left[S_p\cap R_{z,1}\neq\emptyset\land S_r\cap R_{z,2}\neq\emptyset\right]
		\end{equation}
	\end{lem}
	\begin{proof}
		We denote by $R_{z,1}(g,u)$ the RR-set rooted at node $u$ under the realization $g\sim\Omega^p$. From the (7), we have $z(S_p)=\mathbb{E}_{g\sim\Omega^p}\mathbb{E}_{g'\sim\Omega^r}[\sum\nolimits_{u\in I_g(S_p)\cap I_{g'}(S_r)}l(u)]=\sum_{u\in V}\Pr_{g\sim\Omega^p,g'\sim\Omega^r}[S_p\cap R_{z,1}(g,u)\neq\emptyset\land S_r\cap R_{z,2}(g',u)\neq\emptyset]\cdot l(u)=l(V)\cdot\sum_{u\in V}\Pr_{g\sim\Omega^p,g'\sim\Omega^r}[S_p\cap R_{z,1}(g,u)\neq\emptyset\land S_r\cap R_{z,2}(g',u)\neq\emptyset]\cdot(l(u)/l(V))=l(V)\cdot\Pr_{g\sim\Omega^p,g'\sim\Omega^r,u}[S_p\cap R_z(g,g',u)\neq\emptyset\land S_r\cap R_z(g,g',u)\neq\emptyset]$. The (12) is establish equivalently.
	\end{proof}

	As mentioned above, we have to generate two collections of RR sets, $\mathcal{R}_w=\{R_w^1,R_w^2,\cdots,R_w^\lambda\}$ to estimate $w(S_p)$ and $\mathcal{R}_z=\{R_z^1,R_z^2,\cdots,R_z^\mu\}$ to estimate $z(S_p)$. Them we define the following two estimations
	\begin{flalign}
		&F_{\mathcal{R}_w}(S_p)=\frac{1}{\lambda}\cdot\sum_{i=1}^{\lambda}\mathbb{I}[S_p\cap R_w^i\neq\emptyset]\\
		&F_{\mathcal{R}_z}(S_p)=\frac{1}{\mu}\cdot\sum_{i=1}^{\mu}\mathbb{I}[S_p\cap R_{z,1}^i\neq\emptyset\land S_r\cap R_{z,2}^i\neq\emptyset]
	\end{flalign}
	the fraction of RR-sets covered by $S_p$ where $\mathbb{I}[\cdot]$ is an indicator such that $\mathbb{I}[S_p\cap R_w^i\neq\emptyset]=1$ if $S_p\cap R_w^i\neq\emptyset]=1$, or else $\mathbb{I}[S_p\cap R_w^i\neq\emptyset]=0$. Then, we have $\hat{w}(S_p)=p(V)\cdot F_{\mathcal{R}_w}(S_p)$, $\hat{z}(S_p)=l(V)\cdot F_{\mathcal{R}_z}(S_p)$, and $\hat{f}(S_p)=\hat{w}(S_p)-\hat{z}(S_p)$. Next, to bound the gap between ground-truth and estimator, we introduce the Chernoff-Hoeffding inequality.
	\begin{lem}[Chernoff-Hoeffding]
		Let $X_1,X_2,\cdots,X_\theta$ be a series of random variables sampled from a distribution $X$ with expectation $\mathbb{E}[X]$ independently and identically in the set $\{0,1\}$. Given an error $\varepsilon>0$, we have
		\begin{flalign}
			&\Pr\left[\sum\nolimits_{i=1}^{\theta}X_i-\theta\cdot\mathbb{E}[X]\geq+\varepsilon\right]\leq\exp\left(-\frac{2\varepsilon^2}{\theta}\right)\\
			&\Pr\left[\sum\nolimits_{i=1}^{\theta}X_i-\theta\cdot\mathbb{E}[X]\leq-\varepsilon\right]\leq\exp\left(-\frac{2\varepsilon^2}{\theta}\right)
		\end{flalign}
	\end{lem}
	According to the Lemma 2, we can get the relationship between $F_{\mathcal{R}_w}(S_p)$ and its real value $w(S_p)$.
	\begin{lem}
		Given a collection of RR-sets $\mathcal{R}_w$ with $|\mathcal{R}_w|=\lambda$ and any $\delta\in(0,4)$, we have
		\begin{flalign}
			&\Pr\left[w(S_p)\geq \hat{w}(S_p)-p(V)\sqrt{\frac{1}{2\lambda}\ln\left(\frac{4}{\delta}\right)}\right]\geq 1-\frac{\delta}{4}\\
			&\Pr\left[w(S_p)\leq \hat{w}(S_p)+p(V)\sqrt{\frac{1}{2\lambda}\ln\left(\frac{4}{\delta}\right)}\right]\geq 1-\frac{\delta}{4}
		\end{flalign}
	\end{lem}
	\begin{proof}
		To the (17), it is equivalent to prove $\Pr[w(S_p)< \hat{w}(S_p)-p(V)\cdot\sqrt{(1/(2\lambda))\ln(4/\delta)}]\leq\delta/4$. Then, we have $\Pr[w(S_p)< p(V)\cdot F_{\mathcal{R}_w}(S_p)-p(V)\cdot\sqrt{(1/(2\lambda))\ln(4/\delta)}]=\Pr[\lambda\cdot F_{\mathcal{R}_w}(S_p)-\lambda\cdot w(S_p)/p(V)>\sqrt{(\lambda/2)\ln(4/\delta)}]\leq\exp(-2\cdot(\lambda/2)\ln(4/\delta)/\lambda)=\delta/4$ based on the (15).
		
		Similarly, to the (18), it is equivalent to prove $\Pr[w(S_p)> \hat{w}(S_p)+p(V)\cdot\sqrt{(1/(2\lambda))\ln(4/\delta)}]\leq\delta/4$. Then, we have $\Pr[w(S_p)> p(V)\cdot F_{\mathcal{R}_w}(S_p)+p(V)\cdot\sqrt{(1/(2\lambda))\ln(4/\delta)}]=\Pr[\lambda\cdot F_{\mathcal{R}_w}(S_p)-\lambda\cdot w(S_p)/p(V)<-\sqrt{(\lambda/2)\ln(4/\delta)}]\leq\exp(-2\cdot(\lambda/2)\ln(4/\delta)/\lambda)=\delta/4$ based on the (16).
	\end{proof}
	Given an unbiased estimator $\hat{w}(S_p)$, an upper bound and a lower bound of $w(S_p)$ can be defined with at least $1-\delta/4$ probability. Given an unbiased estimator $\hat{w}(S_p)$, an upper bound and a lower bound of $w(S_p)$ can be defined with at least $1-\delta/4$ probability. That is
	\begin{flalign}
		w_u(S_p)&=\hat{w}(S_p)+p(V)\cdot\sqrt{(1/(2\lambda))\ln({4}/{\delta})}\\
		w_l(S_p)&=\hat{w}(S_p)-p(V)\cdot\sqrt{(1/(2\lambda))\ln({4}/{\delta})}
	\end{flalign}
	Given a collection of RR-sets $\mathcal{R}_z$ with $|\mathcal{R}_z|=\mu$, any $\delta\in(0,4)$, and an unbiased estimator $\hat{z}(S_p)$, an upper bouand and a lower bound of $z(S_p)$ can be defined at least $1-\delta/4$ probability in the same way. That is
	\begin{flalign}
		z_u(S_p)&=\hat{z}(S_p)+l(V)\cdot\sqrt{(1/(2\mu))\ln({4}/{\delta})}\\
		z_l(S_p)&=\hat{z}(S_p)-l(V)\cdot\sqrt{(1/(2\mu))\ln({4}/{\delta})}
	\end{flalign}
	Based on the (19)$-$(21), we can derive a lower bound for our objective value $f(S_p)$ naturally. 
	\begin{lem}
		Given any seed set $S_p\subseteq V$, we can take $w_u(S_p)-z_l(S_p)$ as an upper bound of $f(S_p)$ with at least $1-\delta/2$ probability and $w_l(S_p)-z_u(S_p)$ as a lower bound of $f(S_p)$ with at least $1-\delta/2$ probability.
	\end{lem}
	\begin{proof}
		To estimate the $f(S_p)$, we have $\Pr[f(S_p)\leq w_u(S_p)-z_l(S_p)]\geq\Pr[(w(S_p)\leq w_u(S_p))\land(z(S_p)\geq z_l(S_p))]=(1-\delta/4)\cdot(1-\delta/4)\geq 1-\delta/2$. Similarly, we have $\Pr[f(S_p)\geq w_l(S_p)-z_u(S_p)]\geq\Pr[(w(S_p)\geq w_l(S_p))\land(z(S_p)\leq z_u(S_p))]=(1-\delta/4)\cdot(1-\delta/4)\geq 1-\delta/2$.
	\end{proof}
	Next, we are going to discuss how to compute the upper bound of our objective value $\pi(S_p^\circ)$ according to the solution $S_p^\circ$ returned by Algorithm \ref{a1}. The value of $\hat{\pi}(S_p)$ can be obtained by $\hat{f}(S_p)$, which has been decomposed as $\hat{f}(S_p)=\hat{w}(S_p)-\hat{z}(S_p)$. Here, $\hat{w}(S_p)$ and $\hat{z}(S_p)$ are monotone and submodular with respect to $S_p$ as well since they can be reduced to the set coverage problem. Therefore, for any set $X$, we have $\hat{\pi}(X)\geq\max_{|Y|\leq k}\hat{f}(Y)$. From the Lemma 4, the objective value $f(S_p)$ is upper bounded by $w_u(S_p)-z_l(S_p)$ with a high probability. Thereby we have the following conclusion.
	\begin{lem}
		Given the solution $S_p^\circ$ returned by Algorithm \ref{a1}, for any seed set $S_p\subseteq V$ and any $\delta\in(0,4)$, we have
		\begin{flalign}
			f(S_p)&\leq\hat{\pi}(S_p^\circ)\nonumber\\
			&+p(V)\sqrt{\frac{1}{2\lambda}\ln\left(\frac{4}{\delta}\right)}+l(V)\sqrt{\frac{1}{2\mu}\ln\left(\frac{4}{\delta}\right)}
		\end{flalign}
		holds with at least $1-2/\delta$ probability.
	\end{lem}
	\begin{proof}
		According to the Lemma 4, we have $\Pr[f(S_p)\leq w_u(S_p)-z_l(S_p)]\geq1-\delta/2$. Then, $f(S_p)\leq w_u(S_p)-z_l(S_p)=\hat{w}(S_p)-\hat{z}(S_p)+p(V)\cdot\sqrt{(1/(2\lambda))\ln({4}/{\delta})}+l(V)\cdot\sqrt{(1/(2\mu))\ln({4}/{\delta})}=\hat{f}(S_p)+p(V)\cdot\sqrt{(1/(2\lambda))\ln({4}/{\delta})}+l(V)\cdot\sqrt{(1/(2\mu))\ln({4}/{\delta})}\leq\hat{\pi}(S_p^\circ)+p(V)\cdot\sqrt{(1/(2\lambda))\ln({4}/{\delta})}+l(V)\cdot\sqrt{(1/(2\mu))\ln({4}/{\delta})}$, which holds with at least $1-\delta/2$ probability.
	\end{proof}
	\begin{thm}
		The approximation guarantee achieved by the solution $S_p^\circ$ returned by Algorithm \ref{a1} satisfies as follows: $f(S_p^\circ)/\max_{|S_p|\leq k}f(S_p)\geq$
		\begin{equation}
			\frac{w_l(S_p^\circ)-z_u(S_p^\circ)}{\hat{\pi}(S_p^\circ)+p(V)\sqrt{\frac{1}{2\lambda}\ln\left(\frac{4}{\delta}\right)}+l(V)\sqrt{\frac{1}{2\mu}\ln\left(\frac{4}{\delta}\right)}}
		\end{equation}
		holds with at least $1-\delta$ probability.
	\end{thm}
	\begin{proof}
		Based on the Lemma 4, we have $f(S_p^\circ)\geq w_l(S_p^\circ)-z_u(S_p^\circ)$ holds with at least $1-\delta/2$ probability. Then based on the Lemma 5, we have $\max_{|S_p|\leq k}f(S_p)\leq\hat{\pi}(S_p^\circ)+p(V)\cdot\sqrt{(1/(2\lambda))\ln({4}/{\delta})}+l(V)\cdot\sqrt{(1/(2\mu))\ln({4}/{\delta})}$ holds with at least $1-\delta/2$ probability. Thereby the approximation (24) is established with at least $1-\delta$ probability.
	\end{proof}
	
\section{Numerical Experiments}
In this section, we carry out several experiments on different datasets to validate the performance of our proposed algorithms. It aims to test the efficiency of modular-modular procedure, shown as Algorithm \ref{a1}, and its effectiveness compared to other heuristic algorithms. All of our experiments are programmed by python, and run on Windows machine with a 3.40GHz, 4 core Intel CPU and 16GB RAM. There are four datasets used in our experiments: (1) NetScience \cite{nr}: a co-authorship network, co-authorship among scientists to publish papers about network science; (2) Wiki \cite{nr}: a who-votes-on-whom network, which comes from the collection Wikipedia voting; (3) Bitcoin \cite{snapnets}: a who-trusts-whom network of people who trade using Bitcoin on a platform called Bitcoin Alpha. The statistics information about these four datasets is represented in Table \ref{table1}. For an undirected graph, each undirected edge is replaced with two reversed directed edges.

\begin{table}[h]
	\renewcommand{\arraystretch}{1.3}
	\caption{The datasets statistics $(K=10^3)$}
	\label{table1}
	\centering
	\begin{tabular}{|c|c|c|c|c|}
		\hline
		\bfseries Dataset & \bfseries n & \bfseries m & \bfseries Type & \bfseries Avg.Degree\\
		\hline
		Netscie & 0.40 K & 1.01 K & undirect & 5.00\\
		\hline
		Wikivot & 1.00 K & 3.15 K & directed & 6.20\\
		\hline
		Bitcoin & 4.00 K & 25.1 K & directed & 12.5\\
		\hline
	\end{tabular}
\end{table}
\subsection{Experimental Settings}
	The diffusion process is based on the IC-model by default. Under the IC-model, we set the diffusion probability $p_{uv}=1/|N^-(v)|$ for each $(u,v)\in E$ as the inverse of $v$'s in-degree, which has been given by many existing researches about the IM problem. For each node $u\in V$, there is a benefit weight and a disturbed wight associated with it. We sample its benefit weight $p(u)$ from $[0,1]$ uniformly and sample its disturbed benefit weight $q(u)$ from $[-1,p(u)]$ uniformly.
	
	\begin{figure*}[!t]
		\centering
		\subfigure[Netscie, modmod-1]{
			\includegraphics[width=0.24\linewidth]{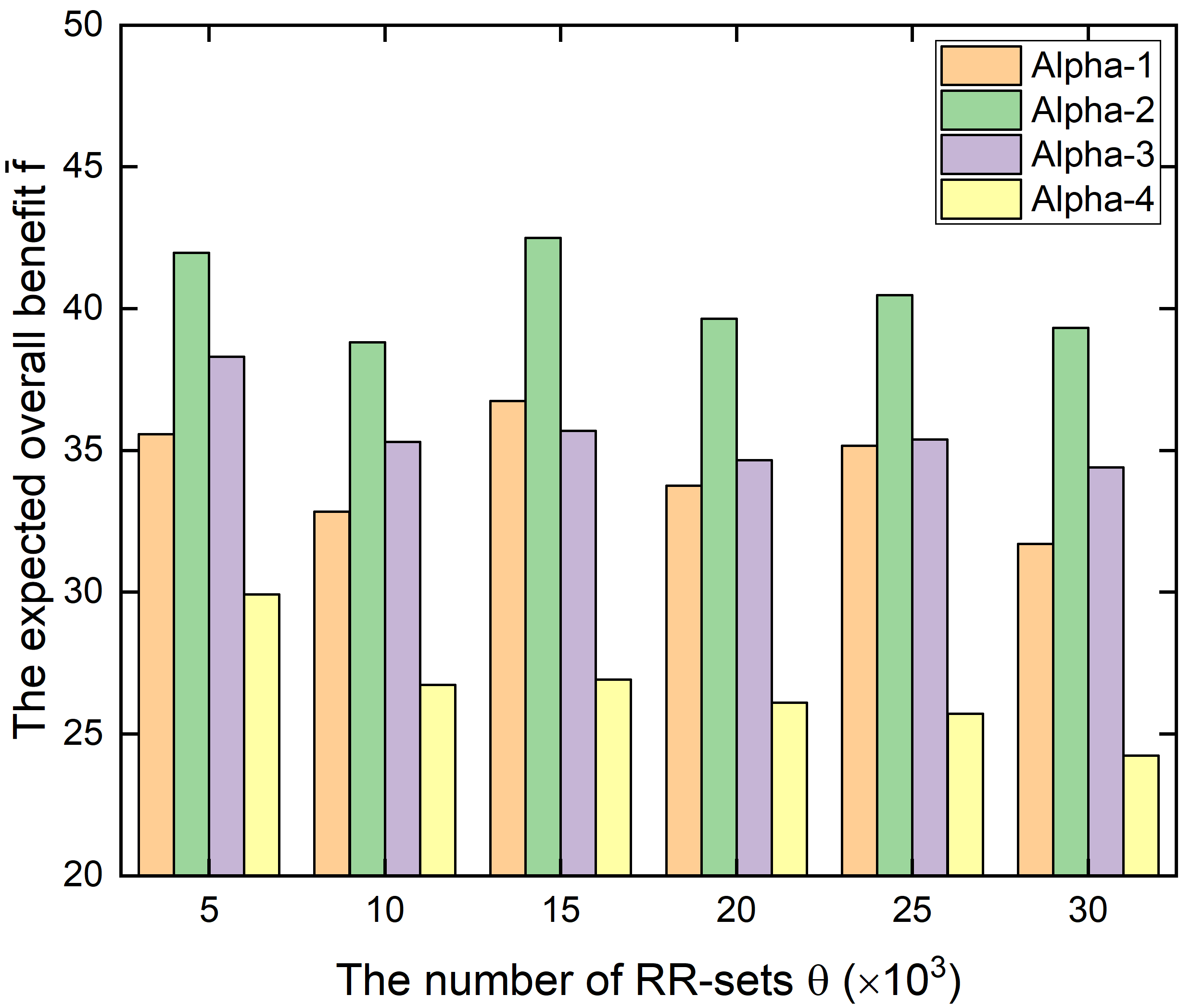}
		}%
		\subfigure[Netscie, modmod-2]{
			\includegraphics[width=0.24\linewidth]{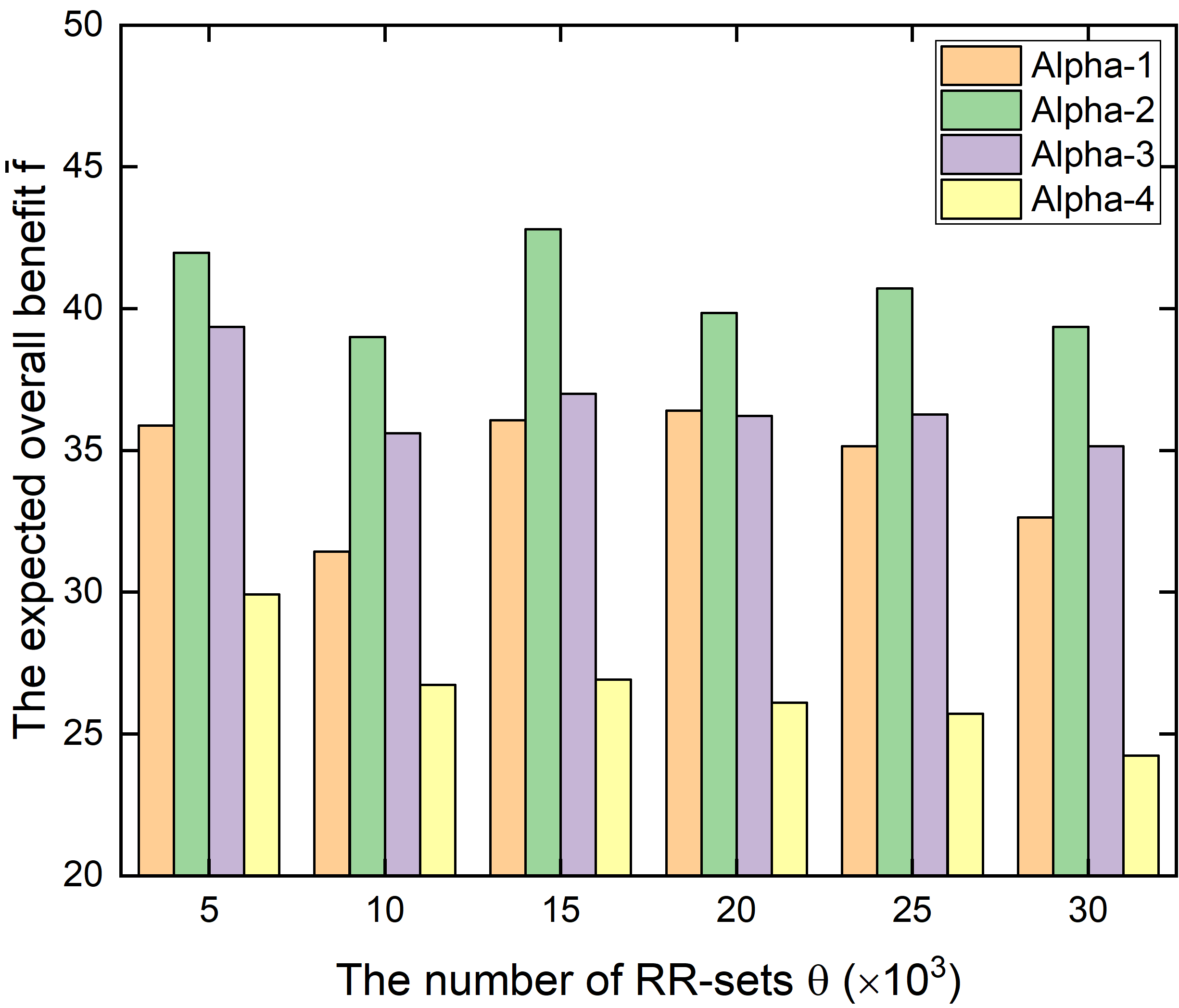}
		}%
		\subfigure[Wikivot, modmod-1]{
			\includegraphics[width=0.24\linewidth]{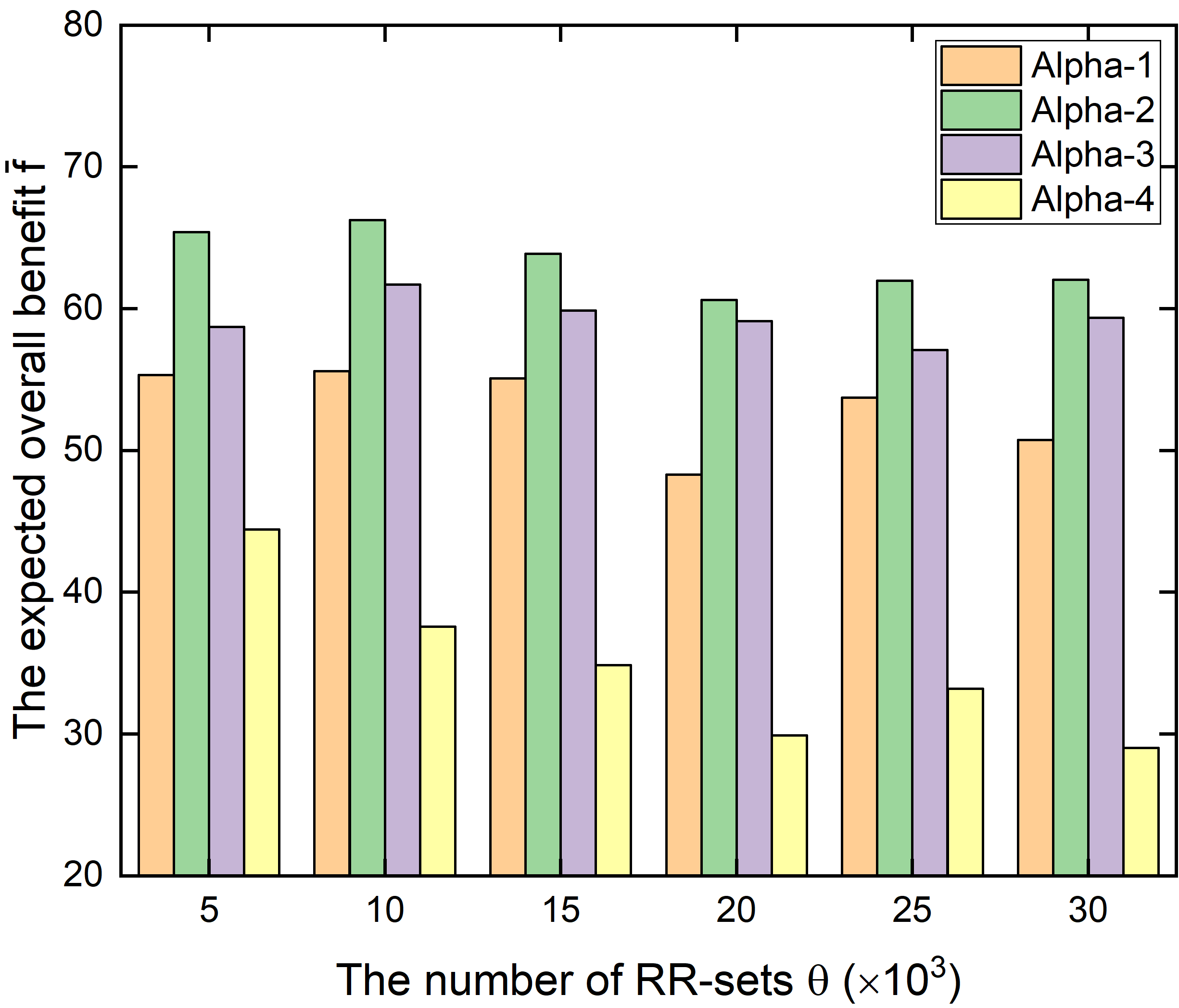}
		}%
		\subfigure[Wikivot, modmod-2]{
			\includegraphics[width=0.24\linewidth]{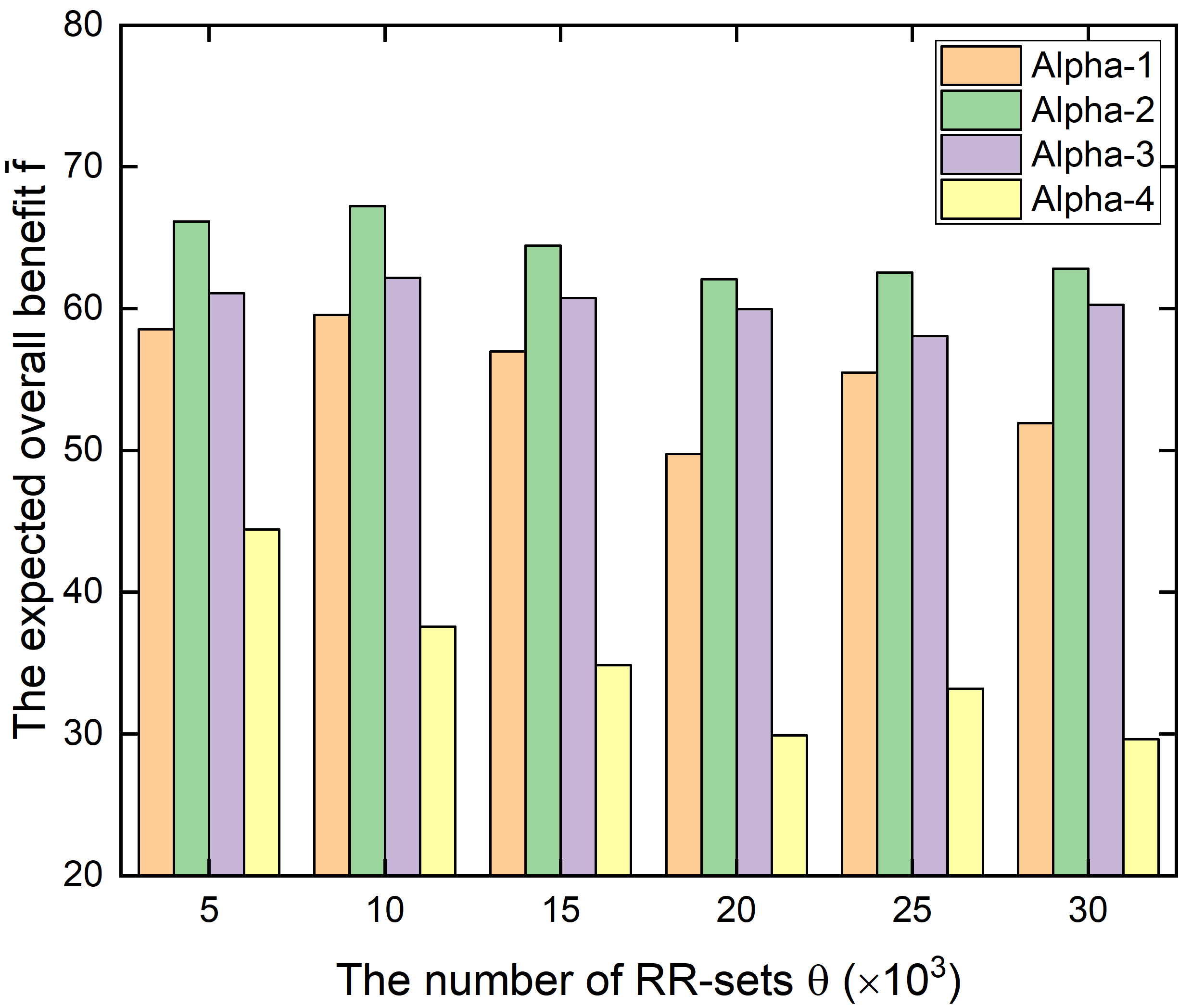}
		}%
		\centering
		\caption{The performance comparison of four permutation selections under the different datasets and upperbounds.}
		\label{fig3}
	\end{figure*}
	\begin{figure*}[!t]
		\centering
		\subfigure[Netscie, $\theta=5 K$]{
			\includegraphics[width=0.24\linewidth]{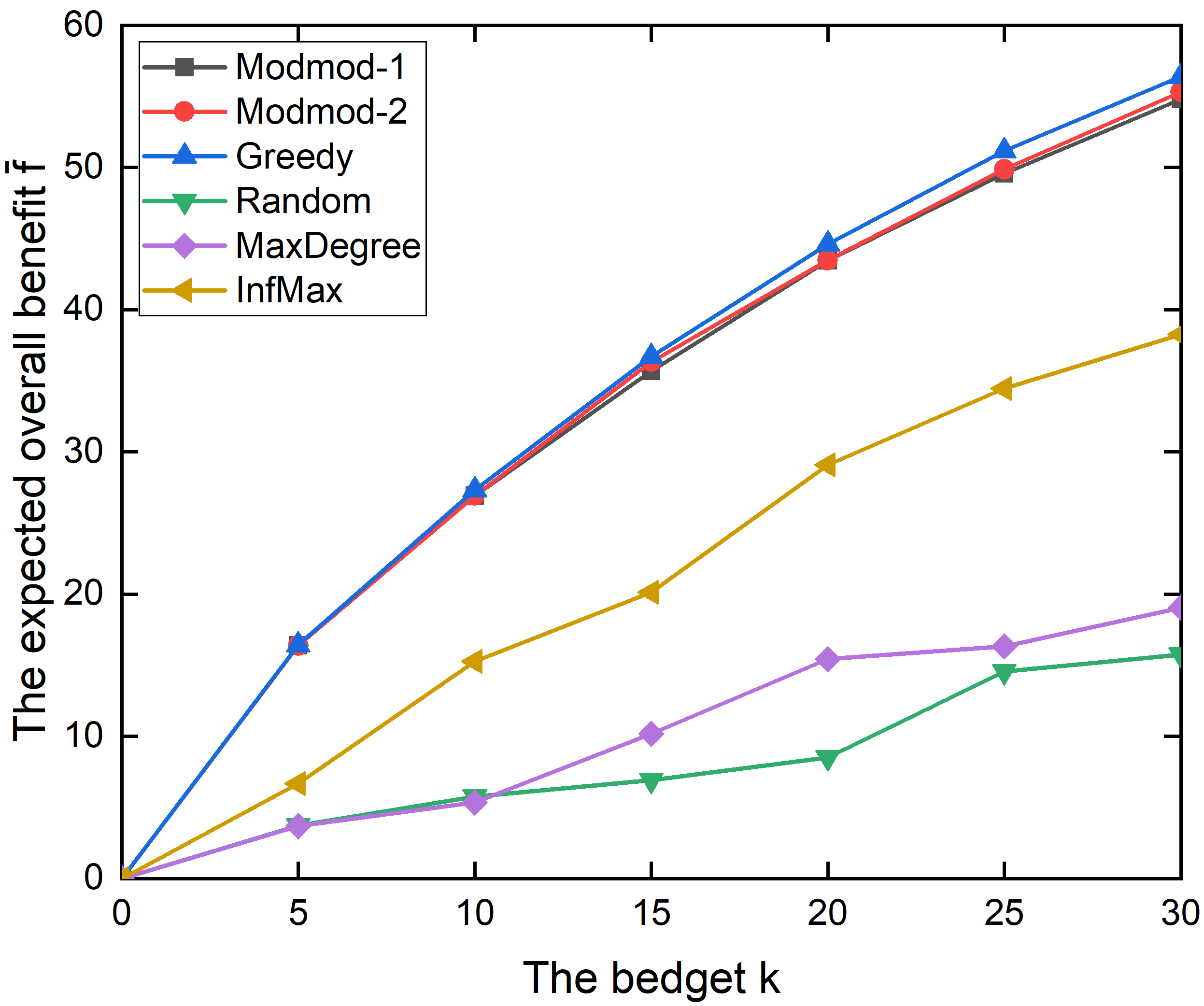}
		}%
		\subfigure[Netscie, $\theta=10 K$]{
			\includegraphics[width=0.24\linewidth]{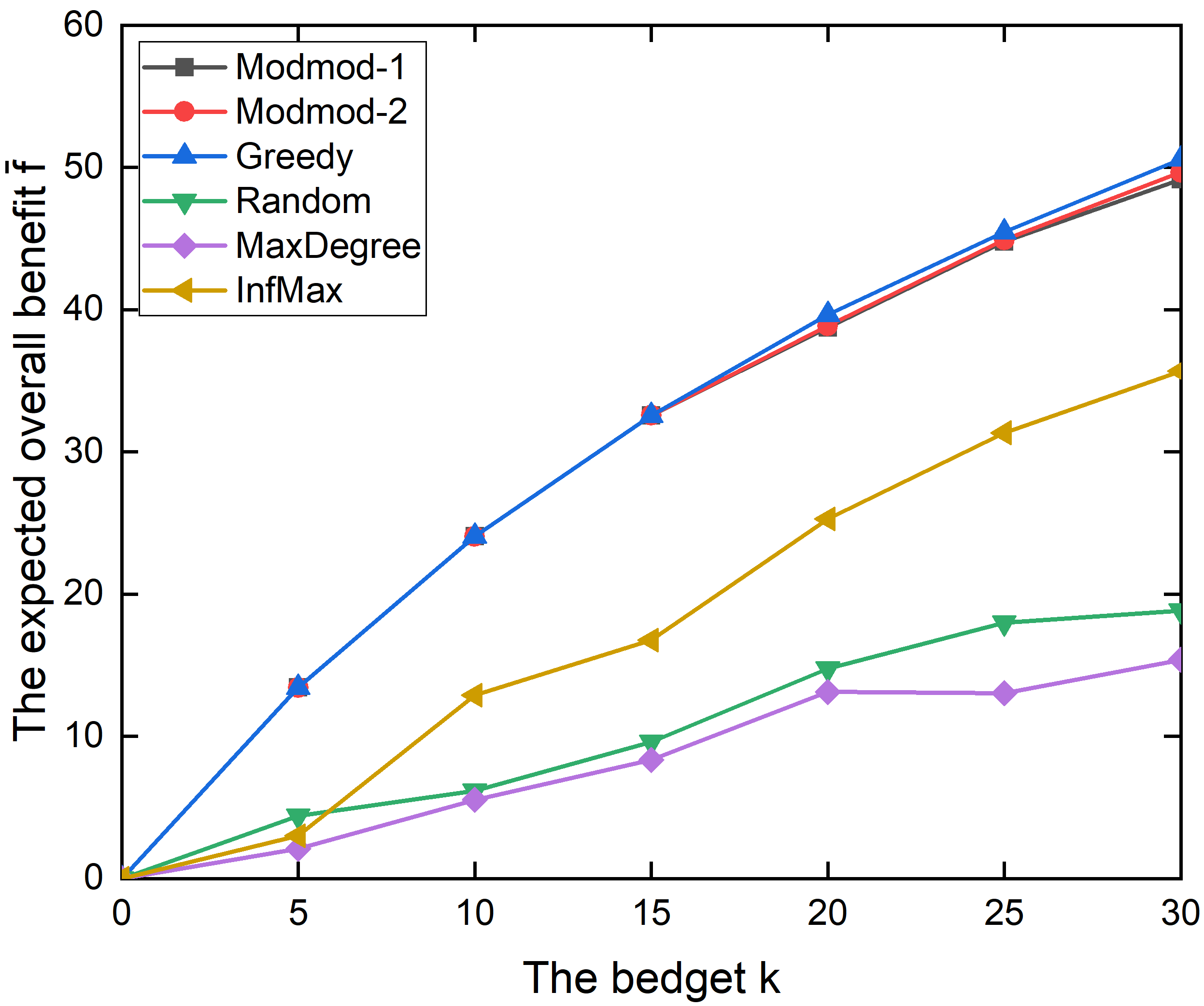}
		}%
		\subfigure[Netscie, $\theta=15 K$]{
			\includegraphics[width=0.24\linewidth]{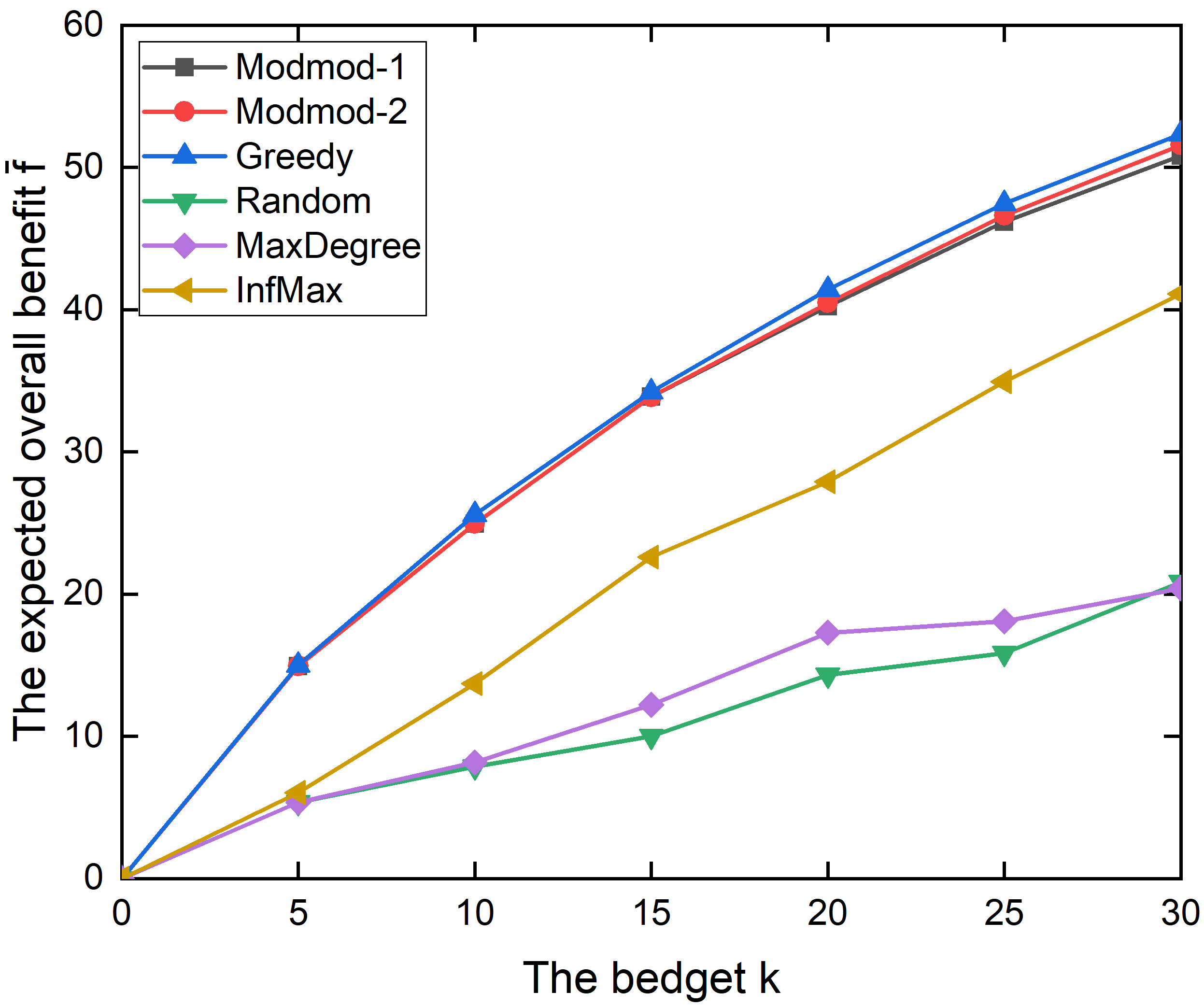}
		}%
		\subfigure[Netscie, $\theta=20 K$]{
			\includegraphics[width=0.24\linewidth]{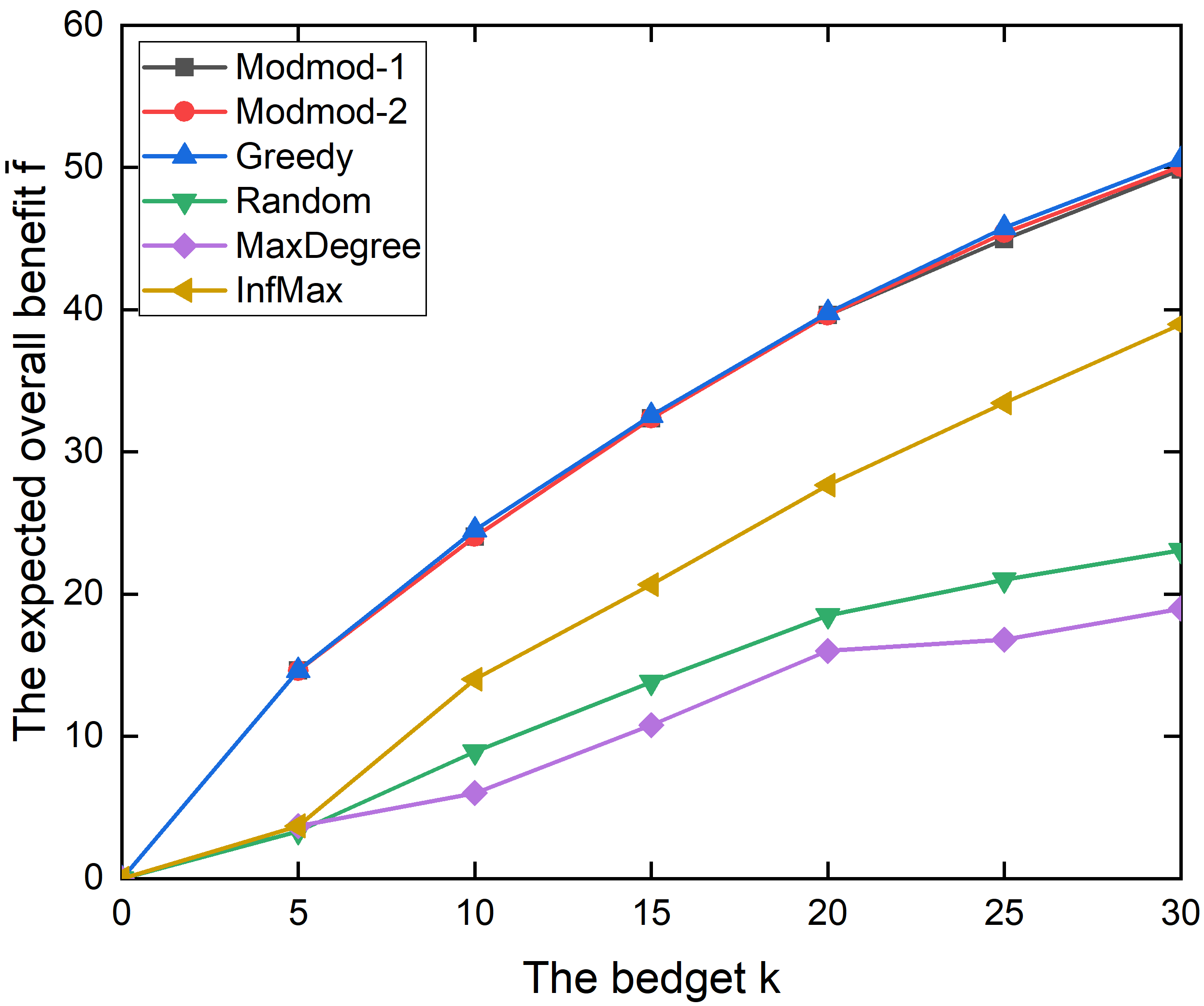}
		}%
		\centering
		\caption{The performance comparison with other heuristic algorithms under the Netscie dataset.}
		\label{fig4}
	\end{figure*}
	
	Consider the modular-modular procedure, we have to define a modular lower bound for the function $w(\cdot)$ and a modular upper bound for the function $z(\cdot)$. Here, we denote ``modmod-1'' to imply that we use the first upper bound $m^z_{X,1}(Y)$ defined in (8) and ``modmod-2'' to imply that we use the second upper bound $m^z_{X,2}(Y)$ defined in (9). Then, we need to compare our modular-modular procedure with other heuristic algorithms, especially for the greedy algorithm. The greedy algorithm is shown in Algorithm \ref{a3}, which selects the node with the maximum marginal expected overall benefit at each iteration until there is no positive marginal gain can be obtained. Other heuristic algorithms are shown as follows: (1) Random: it selects $k$ nodes uniformly from the node set; (2) MaxDegree: it selects $k$ nodes with the largest out-degree; and (3) InfMax: it is similar to the greedy algorithm, but substitutes the overall benefit $f(\cdot)$ with benefit $w(\cdot)$. They are all estimated on the same group of RR-sets, where the number of random RR-set $R_w$ and $R_z$ is denoted by $\theta=\lambda=\mu$.
	
	\begin{algorithm}[h]
		\caption{\text{Greedy}}\label{a3}
		\begin{algorithmic}[1]
			\renewcommand{\algorithmicrequire}{\textbf{Input:}}
			\REQUIRE A set function $f:2^V\rightarrow\mathbb{R}$
			\STATE Initialize: $S_p\leftarrow\emptyset$
			\FOR {$i=1$ to $k$}
			\STATE Select $u^*$ such that $u^*\in\arg\max_{u\in V\backslash S_p}f(u|S_p)$
			\IF {$f(u^*|S_p)<0$}
			\STATE Break
			\ENDIF
			\STATE $S_p\leftarrow S_p\cup\{u^*\}$
			\ENDFOR
			\RETURN $S_p$
		\end{algorithmic}
	\end{algorithm}

	To get a lower bound, the optimal permutation selections is very hard, thus we give several heuristic strategies to get that efficiently. For the permutation $\alpha^t$ that contains $X^t$ at each iteration, there are four heuristic selection strategies. They are (1) Alpha-1: rearrange $X^t$ and $V\backslash X^t$ randomly and respectively, and then concatenate them together as a $\alpha^t$; (2) Alpha-2: sort $X^t$ and $V\backslash X^t$ respectively from largest to smallest according to the expected overall benefit $f(u)$ for each $u\in V$, and then concatenate them together as a $\alpha^t$; (3) Alpha-3: sort $X^t$ and $V\backslash X^t$ respectively from largest to smallest according to the expected benefit $w(u)$ for each $u\in V$, and then concatenate them together as a $\alpha^t$; and (4) Alpha-4: sort $X^t$ and $V\backslash X^t$ respectively from smallest to largest according to the $z(u)$ for each $u\in V$, and then concatenate them together as a $\alpha^t$.
	
\subsection{Experimental Results}
	\textit{1) Permutation selections:} Fig. \ref{fig3} shows the performance comparison of modular-modular procedure under the aforementioned four permutation selections. Shown as Fig. \ref{fig3}, the solution achieved under the Alpha-2 that permutates according to the expected overall benefit has the best performance. Thus, in the follow-up experiments, we default that modular-modular procedure is implemented under the Alpha-2. The performance under the Alpha-3 is slightly worse that under the Alpha-2. The performance under the Alpha-4 is extremely worse, which implies this heuristic selection is invalid. Moreover, the random permutation selection Alpha-1 is unstable, which is sometimes good sometimes bad.
	
		\begin{figure*}[!t]
		\centering
		\subfigure[Wikivot, $\theta=5 K$]{
			\includegraphics[width=0.24\linewidth]{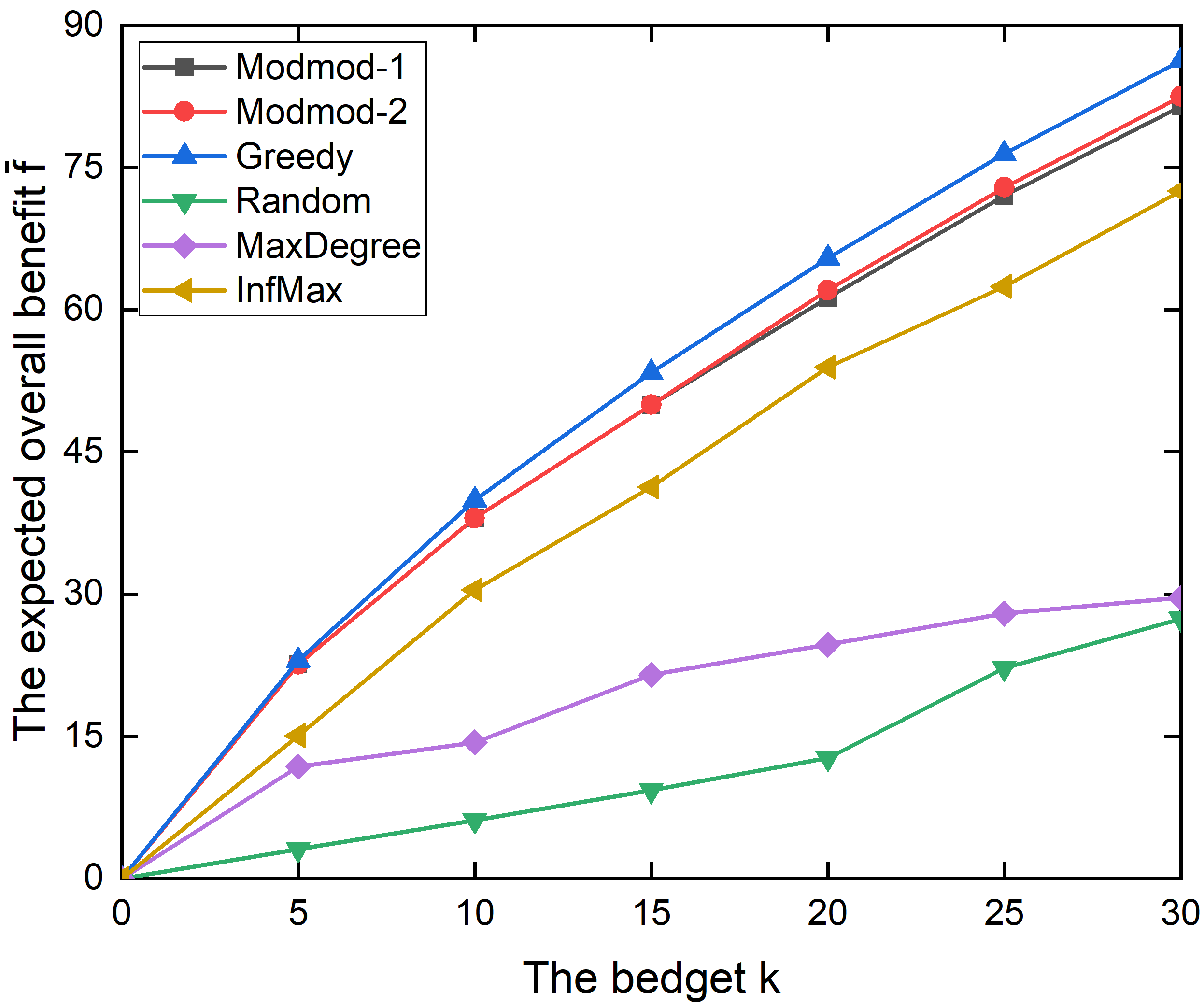}
		}%
		\subfigure[Wikivot, $\theta=10 K$]{
			\includegraphics[width=0.24\linewidth]{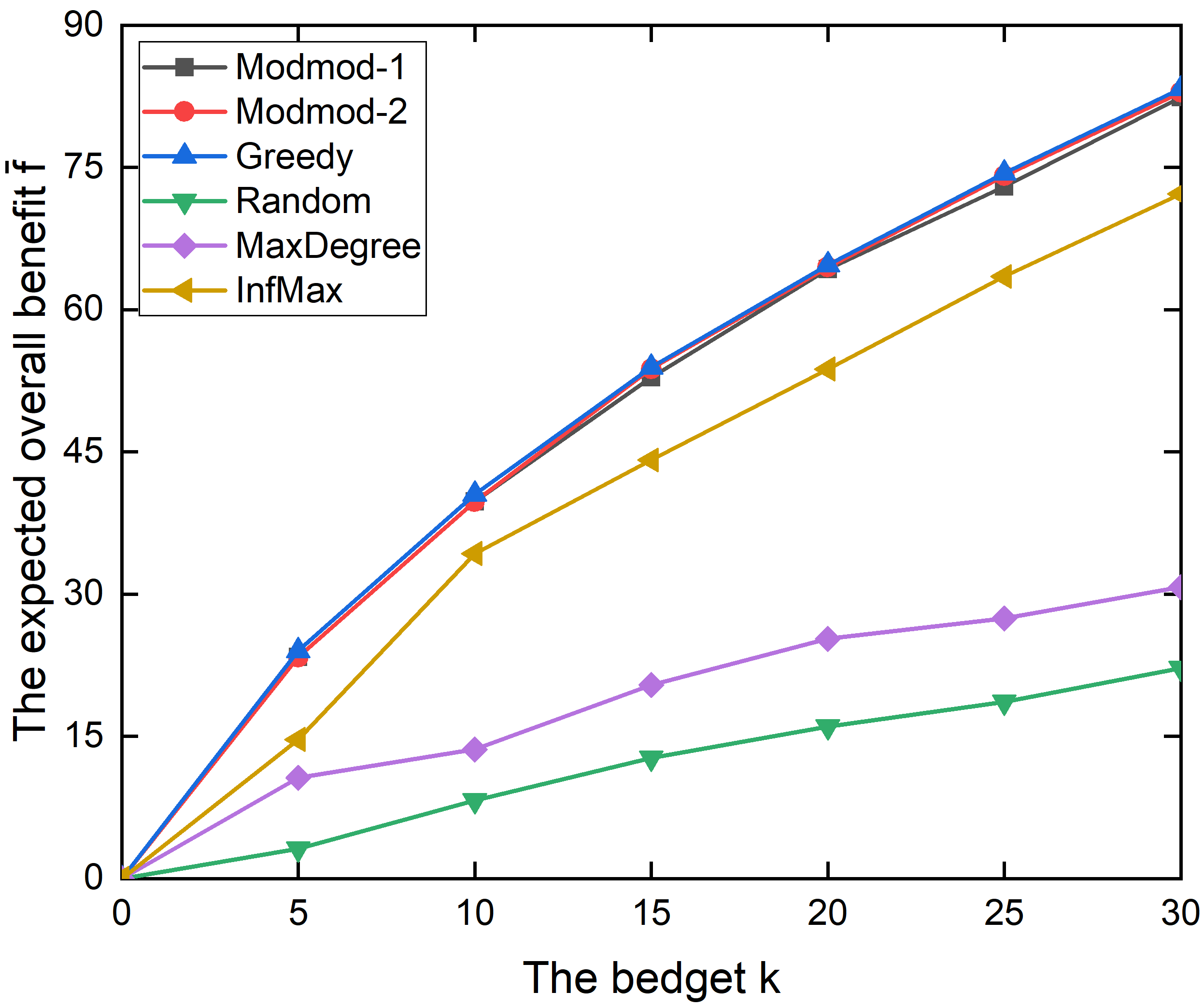}
		}%
		\subfigure[Wikivot, $\theta=15 K$]{
			\includegraphics[width=0.24\linewidth]{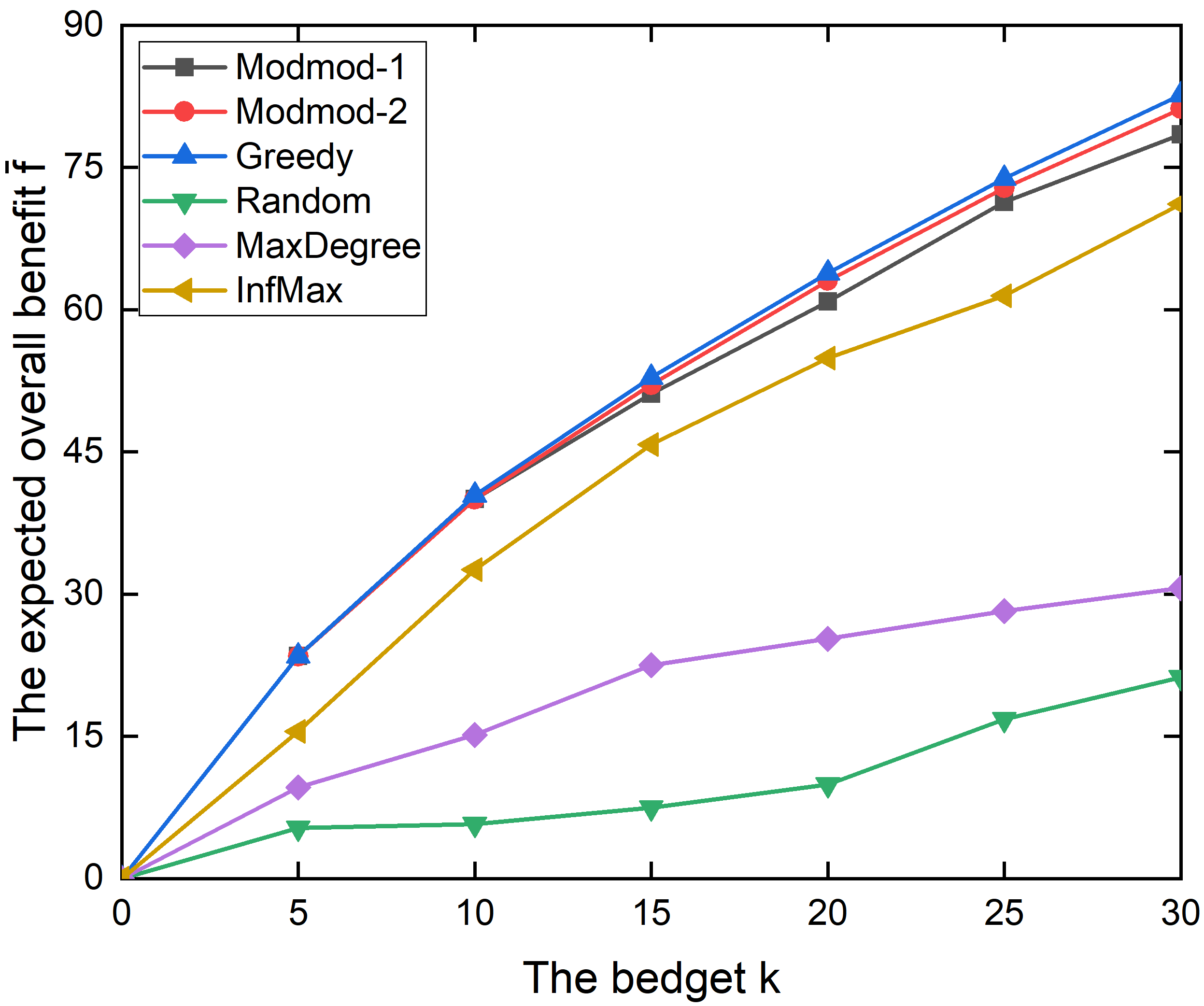}
		}%
		\subfigure[Wikivot, $\theta=20 K$]{
			\includegraphics[width=0.24\linewidth]{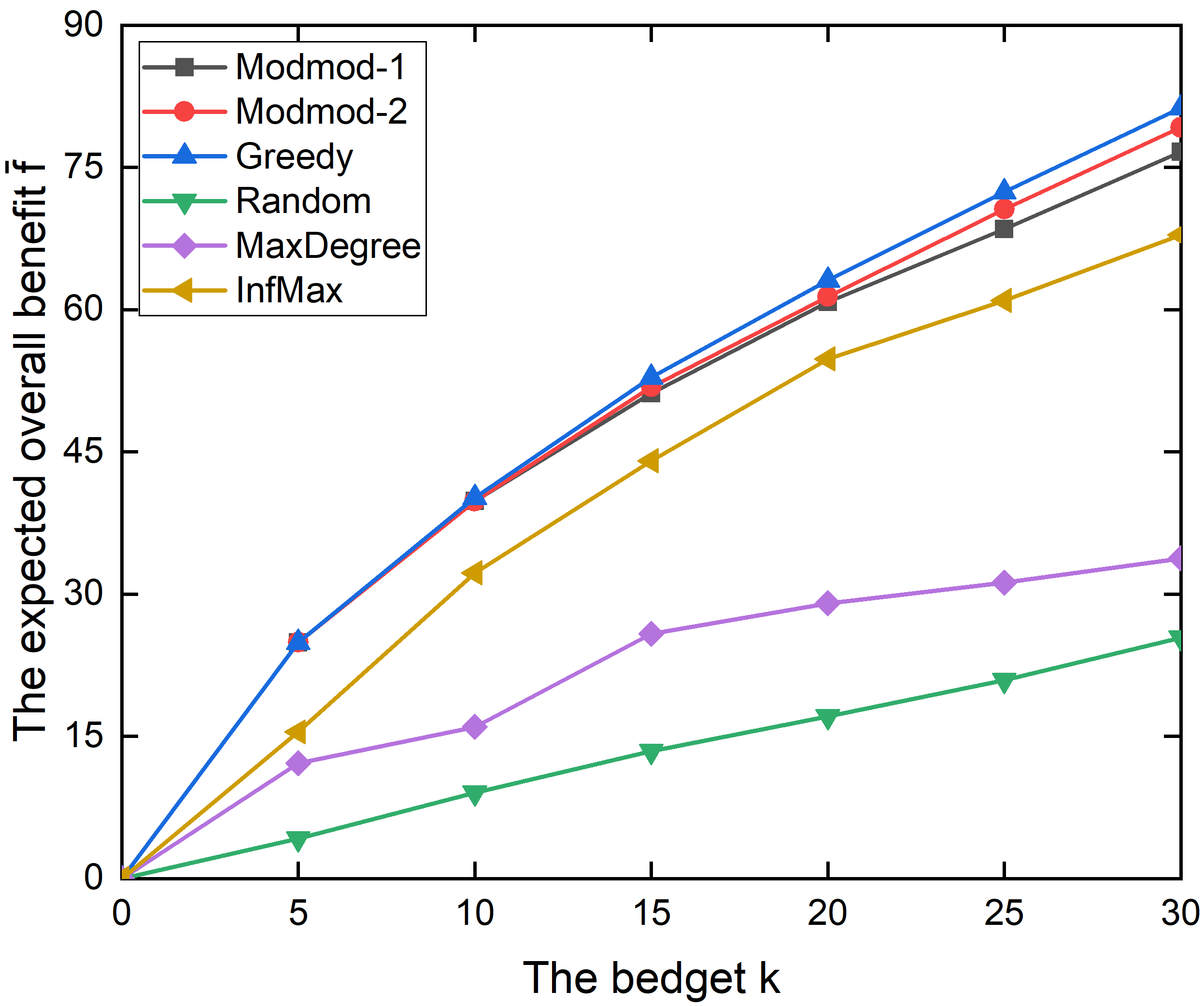}
		}%
		\centering
		\caption{The performance comparison with other heuristic algorithms under the Wikivot dataset.}
		\label{fig5}
	\end{figure*}
	\begin{figure*}[!t]
		\centering
		\subfigure[Bitcoin, $\theta=5 K$]{
			\includegraphics[width=0.24\linewidth]{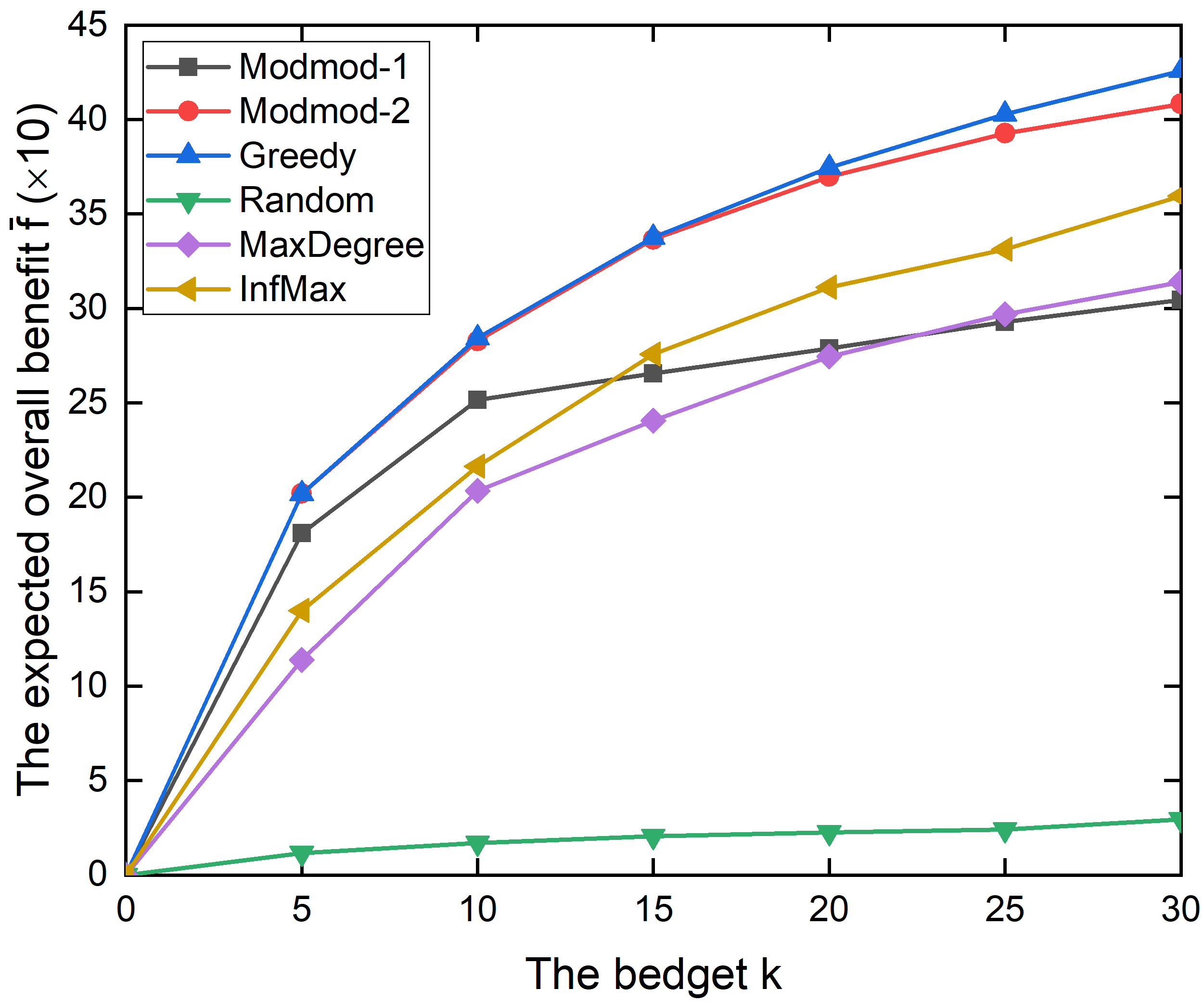}
		}%
		\subfigure[Bitcoin, $\theta=10 K$]{
			\includegraphics[width=0.24\linewidth]{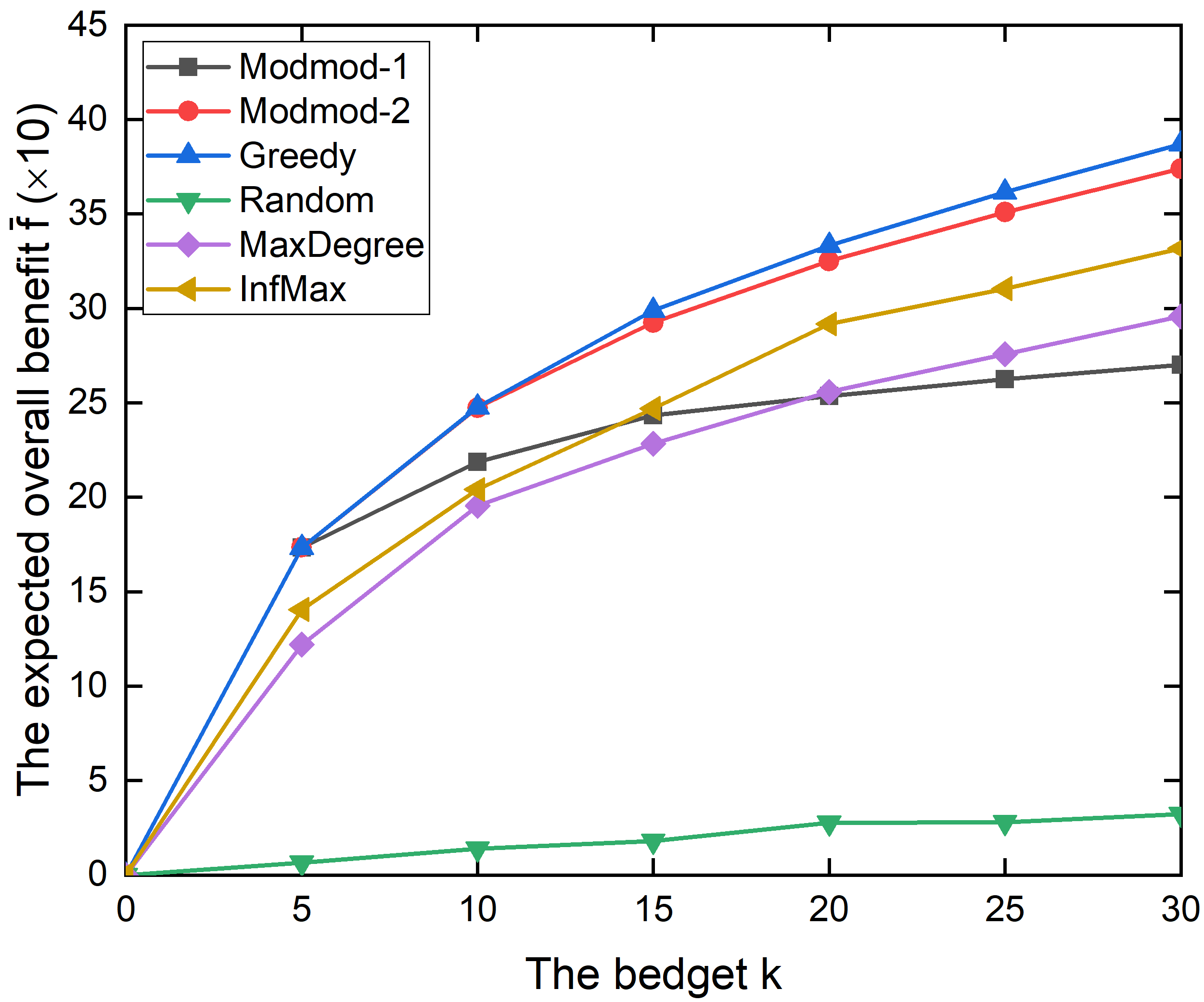}
		}%
		\subfigure[Bitcoin, $\theta=15 K$]{
			\includegraphics[width=0.24\linewidth]{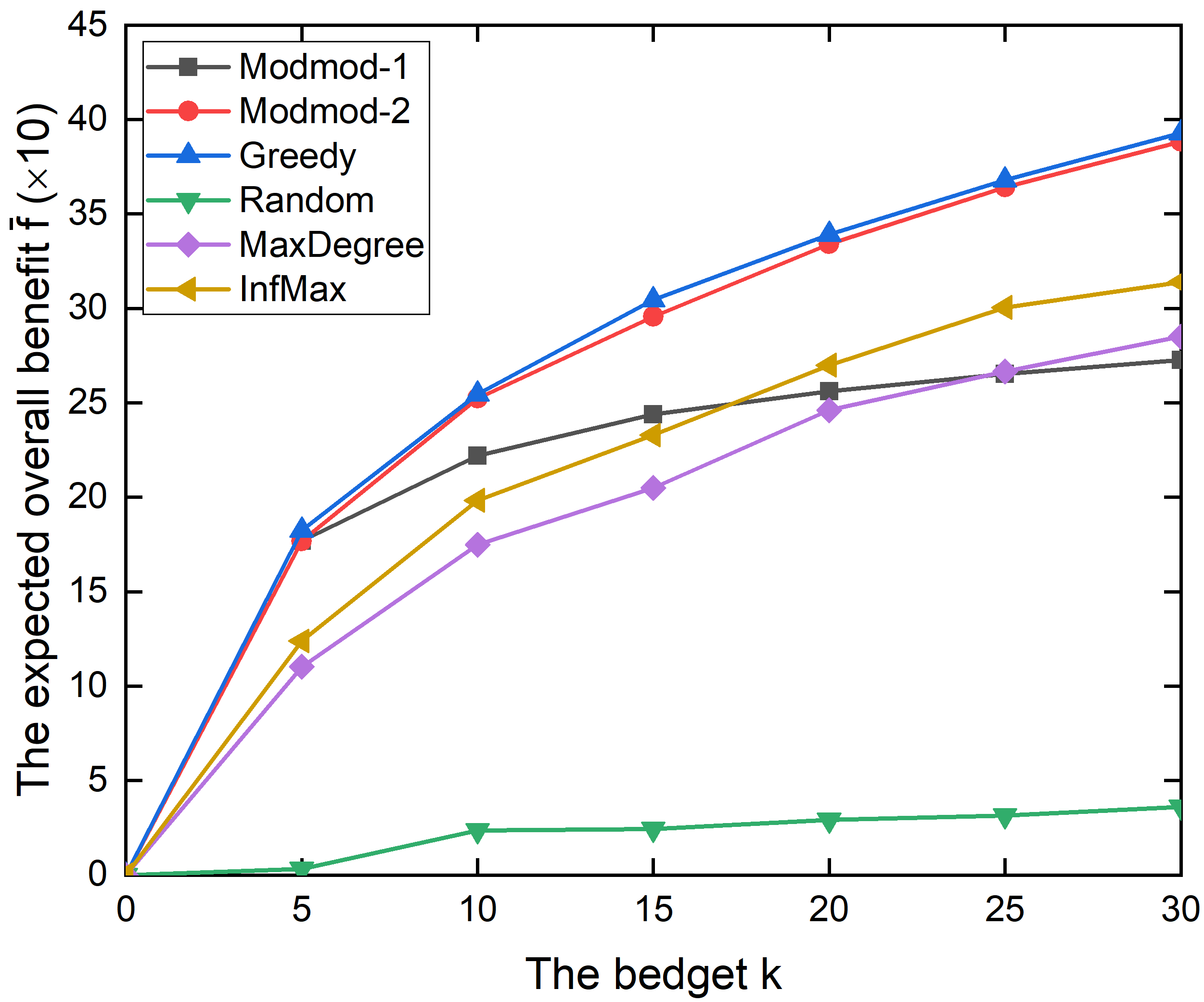}
		}%
		\subfigure[Bitcoin, $\theta=20 K$]{
			\includegraphics[width=0.24\linewidth]{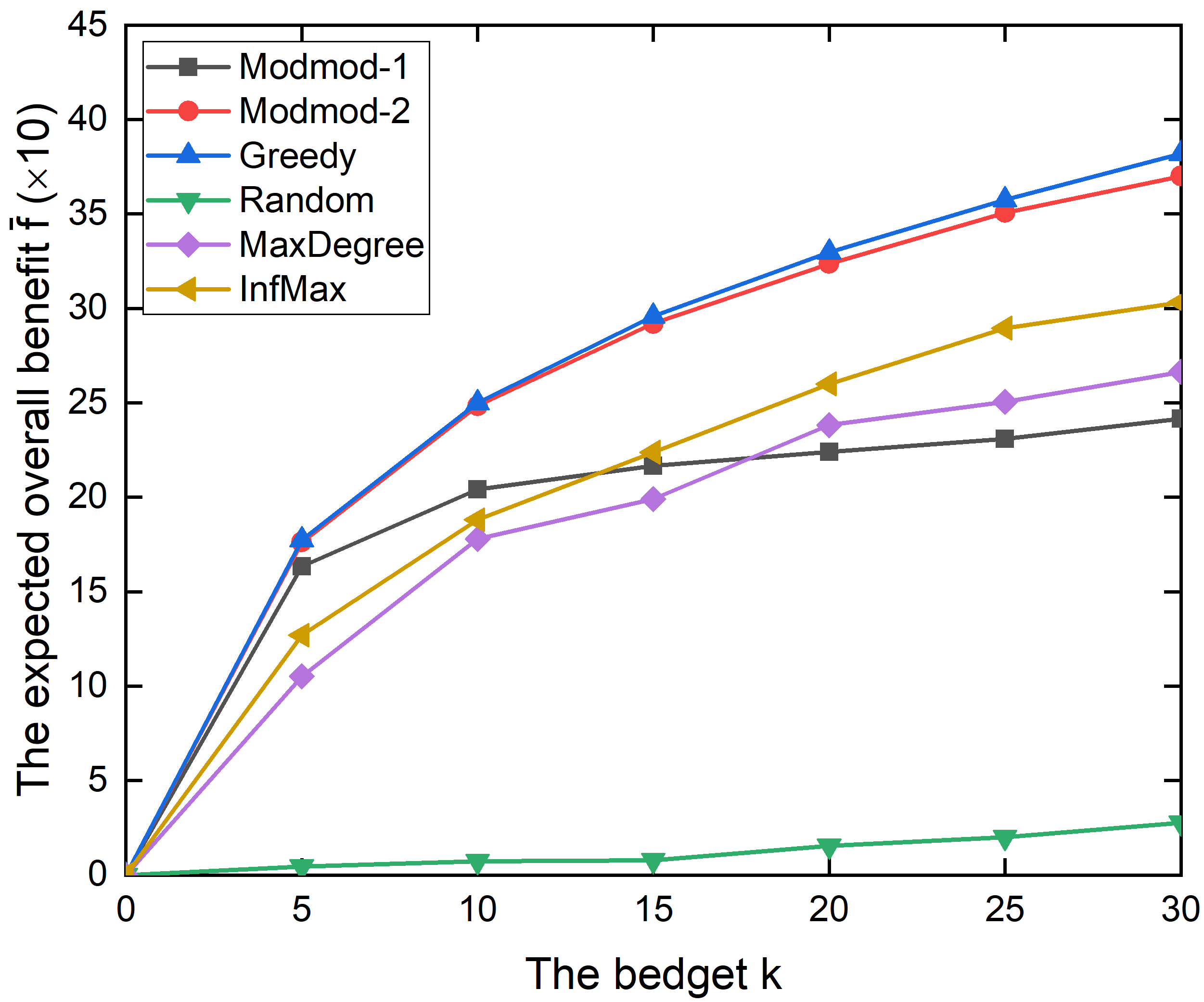}
		}%
		\centering
		\caption{The performance comparison with other heuristic algorithms under the Bitcoin dataset.}
		\label{fig6}
	\end{figure*}

	\begin{table}[!t]
		\renewcommand{\arraystretch}{1.3}
		\caption{Approximation of modular-modular proceduce when $k=20$}
		\label{table2}
		\centering
		\begin{tabular}{|c|c|c|c|c|c|c|}
			\hline
			\bfseries & \multicolumn{2}{|c|}{Netscie} & \multicolumn{2}{|c|}{Wikivot} & \multicolumn{2}{|c|}{Bitcoin} \\
			\hline
			\bfseries $\theta$ & md-1 & md-2 & md-1 & md-2 & md-1 & md-2 \\
			\hline
			5 K & 0.51 & 0.51 & 0.44 & 0.44 & 0.31 & 0.41 \\
			\hline
			10K & 0.50 & 0.50 & 0.47 & 0.47 & 0.31 & 0.42 \\
			\hline
			15K & 0.50 & 0.51 & 0.50 & 0.50 & 0.32 & 0.42 \\
			\hline
			20K & 0.52 & 0.53 & 0.51 & 0.51 & 0.32 & 0.45 \\
			\hline
		\end{tabular}
	\end{table}
	\begin{table}[!t]
		\renewcommand{\arraystretch}{1.3}
		\caption{Running time of modular-modular proceduce when $k=20$}
		\label{table3}
		\centering
		\begin{tabular}{|c|c|c|c|c|c|c|}
			\hline
			\bfseries & \multicolumn{2}{|c|}{Netscie} & \multicolumn{2}{|c|}{Wikivot} & \multicolumn{2}{|c|}{Bitcoin} \\
			\hline
			\bfseries $\theta$ & md-1 & md-2 & md-1 & md-2 & md-1 & md-2 \\
			\hline
			5 K & 09 & 28 & 24 & 083 & 255 & 0935 \\
			\hline
			10K & 17 & 53 & 44 & 154 & 232 & 1190 \\
			\hline
			15K & 23 & 64 & 65 & 410 & 445 & 2587 \\
			\hline
			20K & 27 & 57 & 82 & 285 & 535 & 2481 \\
			\hline
		\end{tabular}
	\end{table}

	\textit{2) Performance of different algorithms: } Fig. \ref{fig4}, Fig. \ref{fig5}, and Fig. \ref{fig6} show the performance comparison with other heuristic algorithms under the different datasets. In these figures, we test the algorithms under the different number of RR-sets. Obviously, the estimations will be more and more accurate as the number of RR-sets increases, but the gap looks inconspicuous from these figures. Then, we have several observations as follows. First, the expected overall benefit increases as the budget increases at least on a budget less than 30. Then, the performances achieved by greedy and modmod-2 algorithms are very close under all datasets. The performances achieved by modmod-1 are unstable under the different datasets, which has good results under the Netscie and Wikivot datasets but a bad result under the Bitcoin dataset. It implies that the selection of upper bound is a critical factor that affects the results of the modular-modular procedure.
	
\subsection{Approximation and Running Time: }
	The approximation and running time of modular-modular procedure when $k=20$ are shown in Table \ref{table2} and Table \ref{table3}. Here, we set the parameter $\delta=0.1$, which means that the approximation ratio shown as \ref{table2} can be satisfied with at least $0.9$ probability. From the Table \ref{table2}, we can see that the approximation ratio improves as the number of RR-sets increases since the estimation errors in (23) can be reduced. From the table \ref{table3}, the running time increases as the number of RR-sets increases generally because the modular maximization process shown as Algorithm \ref{a2} is more time-consuming. However, it is still uncertain since the number of iterations varies under different circumstances, where modmod-2 needs to update $X^t$ more times than modmod-1.

\section{Conclusions}
	In this paper, we consider the disturbance of rival's influence on our benefits we can get from the social networks and propose an OEBI problem formally, which is a generalization for a number of realistic scenarios. Then, we quantify this disturbance, define its objective function, and show its properties. To solve it, we decompose it into the difference of two submodular functions and apply modular-modular procedure to get a solution according to their lower bound and upper bound. Then, we design an efficient unbiased estimate to approximate it with a data-dependent approximation guarantee but reduce running time significantly. These results are verified by numerical simulations based on real-world datasets.

\section*{Acknowledgment}

This work is partly supported by National Science Foundation under grant 1747818 and 1907472.

\ifCLASSOPTIONcaptionsoff
  \newpage
\fi

\bibliographystyle{IEEEtran}
\bibliography{references}

\begin{IEEEbiography}[{\includegraphics[width=1in,height=1.25in,clip,keepaspectratio]{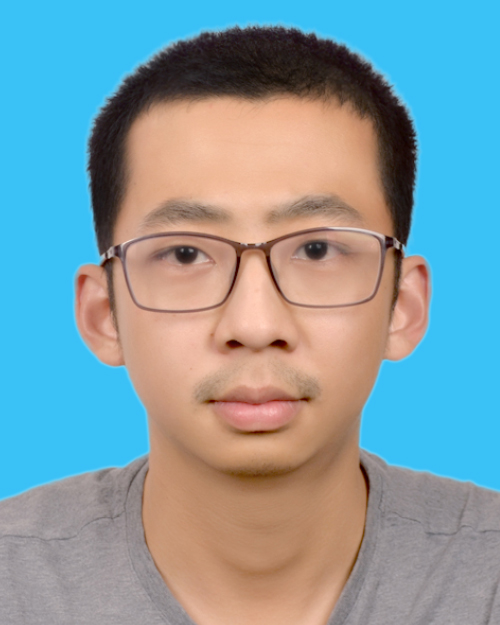}}]{Jianxiong Guo}
	is a Ph.D. candidate in the Department of Computer Science at the University of Texas at Dallas. He received his B.S. degree in Energy Engineering and Automation from South China University of Technology in 2015 and M.S. degree in Chemical Engineering from University of Pittsburgh in 2016. His research interests include social networks, data mining, IoT application, blockchain, and combinatorial optimization.
\end{IEEEbiography}

\begin{IEEEbiography}[{\includegraphics[width=1in,height=1.25in,clip,keepaspectratio]{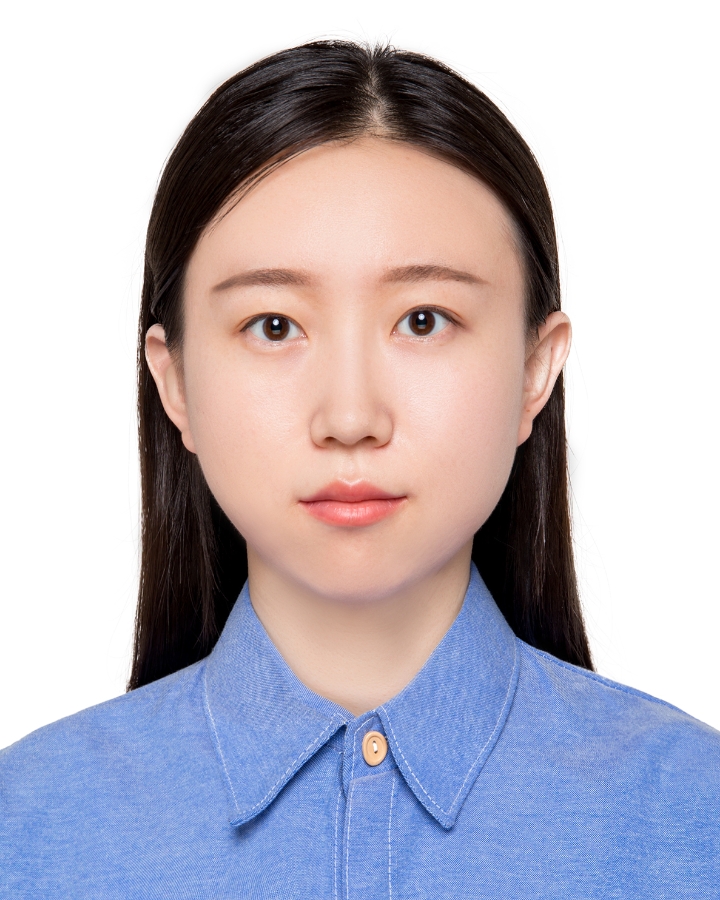}}]{Yapu Zhang}
	received the B.S. degree in Mathematics and Applied Mathematics from Northwest University, Xi'an, China, in 2016. She is a Ph.D. candidate in the School of Mathematical Sciences, University of Chinese Academy of Sciences, Beijing, China. Her research interests include social networks and approximation algorithms.
\end{IEEEbiography}

\begin{IEEEbiography}[{\includegraphics[width=1in,height=1.25in,clip,keepaspectratio]{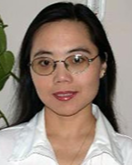}}]{Weili Wu}
	received the Ph.D. and M.S. degrees from the Department of Computer Science, University of Minnesota, Minneapolis, MN, USA, in 2002 and 1998, respectively. She is currently a Full Professor with the Department of Computer Science, The University of Texas at Dallas, Richardson, TX, USA. Her research mainly deals in the general research area of data communication and data management. Her research focuses on the design and analysis of algorithms for optimization problems that occur in wireless networking environments and various database systems.
\end{IEEEbiography}

\end{document}